\documentclass[a4paper,fleqn]{cas-sc}

\usepackage[authoryear,longnamesfirst]{natbib}
\usepackage[ruled,vlined,linesnumbered]{algorithm2e}
\usepackage{amsmath}
\usepackage{amssymb}
\usepackage{amsthm}
\usepackage{graphicx}
\usepackage{wrapfig}
\usepackage{tikz}
\usepackage{tikz-3dplot}
\usetikzlibrary{3d}
\usepackage{amsfonts}
\usepackage{xcolor}
\usepackage{bookmark}
\usepackage{booktabs} 
\usepackage{multirow} 
\usepackage{subfig}
\usepackage{caption}
\usepackage{boxedminipage}
\usepackage{algpseudocode}
\usepackage{colortbl}
\usepackage{url}
\usepackage{bm}
\usepackage{comment}

\newcommand*{\semantics}[1]{\textnormal{[\kern-.15em[}#1\textnormal{]\kern-.15em]}}
\newtheorem{definition}{Definition}
\newtheorem{proposition}{Proposition}
\newtheorem{theorem}{Theorem}

\newcommand{\IE}{\textsf{IE}}

\newcommand{\Pset}{\mathsf{P}}

\newcommand{\Qge}{\mathbb{Q}^{\geq0}}
\newcommand{\Qpos}{\mathbb{Q}^{+}}
\newcommand{\mdiamond}{\Diamond}
\newcommand{\boldBox}{\Box}

\begin{document}
\let\WriteBookmarks\relax
\def\floatpagepagefraction{1}
\def\textpagefraction{.001}

\shorttitle{}    

\shortauthors{Huang et al.}  

\title [mode = title]{A Logic-based Temporal Cohort Discovery Engine:
Algorithms, Indices, and Experimental Results on the 
National Sleep Research Resource}

\author[1]{Yan Huang}
\ead{Yan.Huang@uth.tmc.edu}

\author[1]{Xiaojin Li}
\ead{Xiaojin.Li@uth.tmc.edu}

\author[1]{Licong Cui}
\ead{Licong.Cui@uth.tmc.edu}

\author[1]{Guo-Qiang Zhang}
\cormark[1]
\ead{Guo-Qiang.Zhang@uth.tmc.edu}

\cortext[cor1]{Corresponding author}

\affiliation[1]{organization={University of Texas Health Science Center at Houston},
            city={Houston},
            state={TX},
            country={USA}}

\begin{abstract}
\textbf{\bf Objective:} Large sleep-study repositories contain rich time-stamped physiological annotations, but cohort discovery is still commonly implemented as ad hoc scripts or scalar-index filters. This limits reproducibility when inclusion criteria depend on precise temporal relationships among sleep stages, respiratory events, arousals, desaturations, and signal-derived features. We present a logic-based temporal cohort discovery engine that brings formal semantics, model checking, specialized indexing, and empirical evaluation into a unified biomedical informatics framework. \textbf{Materials and Methods:} We adopt Rational Ensemble Logic (QEL) as a dense-time formal foundation for sleep-data querying and represent each annotated polysomnogram as a Biomedical Event Structure Temporal Model (BEST), a finite mapping from event labels to non-overlapping rational interval ensembles. Cohort discovery is formulated as model checking of QEL formulas over BEST databases. We organize common sleep-research requirements into three reusable temporal query patterns: single-event retrieval, dual-event temporal pattern matching, and event data extraction. To make model checking practical at repository scale, we introduce two indexing strategies: 2-Dimensional Fractional Cascading (2DFC) for interval-overlap retrieval and Fully Connected Fractional Cascading (FCFC) for dual-event pattern matching. \textbf{Results:} The prototype cohort discovery engine was implemented in Python with in-memory and MongoDB-backed execution modes and evaluated on synthetic interval datasets containing up to 90 million intervals and on real-world National Sleep Research Resource annotations from the Cleveland Children's Sleep and Health Study (CCSHS) containing 515 subjects, 202,587 intervals, 23 event labels. 2DFC constructs indexes in linear space and linear build time, reducing build time at 90 million intervals from 11,549 s with RTFC and 23,902 s with 2DRT to 3,655 s. FCFC replaces repeated flatten-sort-scan operations with precomputed boundary-distance arrays and reduced RAM query time by 80-98\% in dual-event experiments; MongoDB experiments showed approximately 50\% reduction for high-candidate workloads. On CCSHS, cohort-selection queries executed at sub-second latency at native scale and under 45 s at $1{,}000\times$ scale. \textbf{Conclusion:} QEL, BEST, and model checking provide mathematically explicit semantics for temporal sleep phenotypes that are both human-readable and machine-computable, while 2DFC and FCFC demonstrate that formal temporal cohort discovery can be executed efficiently on large biomedical data repositories. The resulting framework supports reusable, explainable cohort definitions aligned with sleep-medicine scoring rules and secondary use of polysomnography data. This work is a part of the ``Symbolic Biomedicine'' program championed by the corresponding author.
\end{abstract}



\maketitle

\section{Introduction}
\label{intro}

Sleep medicine is a data-intensive domain in which clinically meaningful phenotypes are often temporal rather than purely scalar. A polysomnogram (PSG) records multi-channel physiological time series over an overnight study and is accompanied by annotations for sleep stages, respiratory events, oxygen desaturations, arousals, limb movements, artifacts, and other scored events. Large repositories such as the National Sleep Research Resource (NSRR~\cite{zhang2018nsrr}) make these recordings available for secondary use, but the dominant cohort-discovery workflow remains a mixture of summary indices, spreadsheet filters, and custom scripts. These approaches are effective for simple variables such as apnea-hypopnea index (AHI) or oxygen desaturation index (ODI), yet they are poorly suited for cohort definitions whose meaning depends on metric temporal relationships: an apnea occurring during N3 sleep, a desaturation beginning within a specified delay after a respiratory event, an arousal overlapping the termination of a hypopnea, or the absence of confounding events during a washout window.

This gap is a biomedical informatics problem. Cohort criteria represent computational phenotype definitions; therefore, temporal criteria should be explicit, reusable, and expressed in a format that is both easily readable by humans and computable by machines. In sleep research, however, temporal logic is rarely treated as a first-class representation. The American Academy of Sleep Medicine (AASM) Scoring Manual~\cite{aasmmanual2015,aasm2012rules} rules and study protocols routinely specify duration thresholds, temporal windows, co-occurrence constraints, precedence relationships, and exclusion intervals, but these semantics are often buried inside analysis code. As a result, two investigators can intend the same clinical definition but obtain different cohorts because of different event anchoring conventions, boundary handling, or delay-window assumptions. This problem is amplified in multi-study repositories where data are reused long after the original scoring context.

We address this problem by developing a logic-based temporal cohort discovery engine for interval-annotated PSG data. The central idea is to formulate cohort discovery as a model-checking task. Each annotated sleep study is represented as a Biomedical Event Structure Temporal Model (BEST), in which every event label is mapped to the finite set of rational intervals during which the event holds. Cohort criteria are written as formulas in Rational Ensemble Logic (QEL~\cite{qel}), 
a dense-time temporal logic in the 
Ensemble Logic spectrum proposed in~\cite{zhang2024temporal} combining first-order quantification, displacement, bounded existence, bounded universality, and Boolean composition in a single framework. A subject belongs to a cohort exactly when the corresponding BEST satisfies the QEL formula at the specified observation point. This gives cohort discovery a precise mathematical semantics and turns temporal phenotype definitions into reusable symbolic artifacts.

The paper brings together five contributions. First, we adopt QEL as a formal foundation for cohort discovery over sleep data, including PSG annotations and signal-derived physiological event intervals. This provides a single logical language for event labels, duration thresholds, bounded temporal windows, and absence constraints. Second, we formulate temporal cohort discovery under the model-checking paradigm by representing annotated PSG recordings as BESTs and by defining query answers as satisfaction sets and temporal-pattern witnesses. Third, we introduce two index strategies that exploit the normalized structure of sleep annotation data: 2-Dimensional Fractional Cascading (2DFC), which accelerates interval-overlap retrieval with linear-space construction, and Fully Connected Fractional Cascading (FCFC), which accelerates dual-event temporal pattern matching through constant-time target-event lookup after preprocessing. Fourth, we implement these ideas in a cohort discovery engine and evaluate the engine on both synthetic datasets and real-world NSRR annotations, including synthetic workloads with up to 90 million intervals and Cleveland Children's Sleep and Health Study (CCSHS) annotations from 515 subjects with 1,000 cohort-query runs and scaled repository experiments. Fifth, we organize sleep-research queries into three temporal query patterns-single-event retrieval, dual-event temporal pattern matching, and event data extraction-that can express practical cohort specifications and rule-like definitions aligned with AASM scoring concepts.

The resulting contribution is not merely a faster interval query algorithm or a new cohort-selection interface. It is an integrated approach that combines mathematical semantics, temporal data representation, model-checking execution, index design, and empirical evaluation for a concrete biomedical repository setting. To our knowledge, no prior biomedical informatics system has combined these components for temporal sleep cohort discovery in a single framework.

\newpage
\noindent\textbf{Statement of Significance.}

\noindent
\begin{tabular}{|p{0.22\textwidth}|p{0.72\textwidth}|}
\hline
\textbf{Problem or Issue} &
Temporal relationships are central to sleep medicine, but they are rarely represented as explicit, reusable, and executable cohort definitions; existing analyses rely on scalar indices, event-pairing heuristics, or study-specific scripts that obscure clinical semantics and limit cross-cohort reproducibility. \\
\hline
\textbf{What is Already Known} &
Sleep repositories such as NSRR store richly annotated interval-based PSG recordings. Temporal database and event-stream systems support useful query operations, yet none provides formally grounded semantics for sleep phenotype definitions that are simultaneously human-readable and machine-computable, nor indexing structures that exploit the sorted, non-overlapping structure of per-label annotation ensembles. \\
\hline
\textbf{What this Paper Adds} &
We introduce a QEL/BEST model-checking framework for reproducible temporal cohort discovery over large sleep archives. QEL supplies a dense-time language for temporal phenotypes; BEST supplies a repository-ready representation of annotated PSGs; and 2DFC/FCFC supply scalable indexing that makes clinically meaningful criteria, including duration thresholds, event co-occurrence, bounded delay, stage restrictions, and signal-level event definitions, available as computational phenotypes that are both human-readable and directly machine-executable. \\
\hline
\textbf{Who Would Benefit} &
Sleep researchers, clinical investigators, data scientists, biomedical informaticists, and platform teams are building cohort-discovery and data-sharing services for secondary use of PSG and other real-world physiological data. \\
\hline
\end{tabular}

\section{Related Work}
\label{sec:queries}

\subsection{Temporal Reasoning on Sleep Data}
Sleep study data combine continuous physiological signals with time-stamped clinical annotations at distinct time scales. Table~\ref{tab:ccshs_annotations} lists annotation event types from the CCSHS~\cite{rosen2003prevalence}.
Sleep stages (N1/N2/N3/REM/WAKE) are scored in fixed 30 s epochs; respiratory events, arousals, SpO\textsubscript{2} desaturations, and limb movements are interval events with onsets and offsets~\cite{aasmmanual2015}. Temporal reasoning over these annotations captures co-occurrence, precedence, and recurrence with more precision than scalar indices (AHI, ODI, arousal index~\cite{tsai1999ahi,chung2012odi,asda1992arousals}).

Many clinically meaningful relationships are defined by timing: hypopneas require associated desaturation or arousal within a specified window~\cite{aasmmanual2015}, and arousals cluster near the termination of obstructive events~\cite{asda1992arousals}. Without explicit temporal operators, such relationships rely on ad-hoc event-pairing heuristics that vary across labs and limit reproducibility across cohorts.

\begin{table}[h]
\centering
\small
\begin{tabular}{llllr}
\hline
Label & Type & Concept & Signal Location & \# of Event \\
\hline
ARASDA\_C3   & Arousals        & ASDA arousal                         & C3         & 42,455 \\
SPO2\_ART    & Respiratory     & SpO2 artifact                        & SpO2       & 36,866 \\
SPO2\_DESAT  & Respiratory     & SpO2 desaturation                    & SpO2       & 31,382 \\
N2           & Stages          & Stage 2 sleep                        & N/A        & 18,000 \\
HYPOP   & Respiratory     & Hypopnea                             & Airflow    & 13,359 \\
WAKE         & Stages          & Wake                                 & N/A        & 13,167 \\
N3           & Stages          & Stage 3 sleep                        & N/A        & 8,530  \\
N1           & Stages          & Stage 1 sleep                        & N/A        & 8,266  \\
LM\_RL1      & Limb Movement   & Limb movement - right                & Right Leg1 & 5,575  \\
LM\_LL1      & Limb Movement   & Limb movement - left                 & Left Leg1  & 5,026  \\
REM          & Stages          & REM sleep                            & N/A        & 4,235  \\
LM\_RL       & Limb Movement   & Limb movement - right                & R Leg      & 3,585  \\
LM\_LL       & Limb Movement   & Limb movement - left                 & L Leg      & 3,453  \\
PLM\_RL      & Limb Movement   & Periodic leg movement - right        & R Leg      & 1,456  \\
PLM\_LL      & Limb Movement   & Periodic leg movement - left         & L Leg      & 1,365  \\
PLM\_RL1     & Limb Movement   & Periodic leg movement - right        & Right Leg1 & 1,071  \\
CA      & Respiratory     & Central apnea                        & Airflow    & 1,024  \\
ARASDA\_C4   & Arousals        & ASDA arousal                         & C4         & 862   \\
PLM\_LL1     & Limb Movement   & Periodic leg movement - left         & Left Leg1  & 776   \\
OA      & Respiratory     & Obstructive apnea                    & Airflow    & 657   \\
\hline
\end{tabular}
\caption{Sleep event annotations with short names in CCSHS.}
\label{tab:ccshs_annotations}
\end{table}

\subsection{Common Absolute and Relative Period Windows Used in Sleep Studies}
Sleep studies operationalize temporal structure via period windows. Absolute windows align to clock-time bands (e.g., 00:00-03:00)~\cite{malicki2022circadian,chen2022dynamic} or fixed post-lights-off/pre-awakening intervals, computing window-normalized event rates (AHI, ODI, arousal index). Relative windows align to the sleep episode's internal structure: SPT partitions for early-vs-late contrasts, or stage-constrained windows (REM vs.\ NREM) for stage-specific burdens and phenotyping~\cite{mokhlesi2012rem,yamauchi2015nrem}. All require explicit choices about anchoring, inclusion rules, and artifact handling for cross-cohort comparability~\cite{aasmmanual2015}.

\subsection{Temporal Queries for Sleep Study Data}
Recurring temporal query patterns in sleep research include precedence, overlap, bounded delay, and stage-specific constraints. Hypoxic burden~\cite{azarbarzin2019hypoxic} moves toward an explicit temporal view but remains study-specific. Key complicating factors, pulse-oximetry delay, epoch-boundary effects, and inter-scorer variability, require queries to encode explicit lags, overlap criteria, and recovery definitions.

\subsection{Existing Temporal Query Systems}
Systems like STAR~\cite{ozcep2014starql}, T-SQL~\cite{microsoft2016tsql}, and KarmaLego~\cite{moskovitch2015karma} incorporate temporal operators or pattern mining over symbolic time intervals; sleep research typically relies on custom scripts. The SQL:2016 \texttt{MATCH\_RECOGNIZE} clause~\cite{iso2016sql} and CEP engines (Esper~\cite{esper2024}, Apache Flink~\cite{carbone2015flink}) apply NFA-based pattern matching over event streams. Automata-based methods require flattening concurrent interval streams into a single sorted sequence, lack metric temporal operators, and cannot exploit the non-overlapping sorted structure of BEST. Our 2DFC and FCFC algorithms address these limitations directly.

\section{Methods}
\label{sec:formal}
The formal development adopts Rational Ensemble Logic over the rational timeline and specializes it to interval-annotated PSG records. The resulting workflow has four layers: (i) raw XML annotations are normalized into labeled intervals; (ii) each subject is represented as a BEST; (iii) QEL formulas define query and cohort criteria; and (iv) model-checking results are accelerated by 2DFC and FCFC indexes. Figure~\ref{fig:bestdb_overview} illustrates the overall pipeline.

\begin{figure}
  \centering
  \includegraphics[width=1\textwidth]{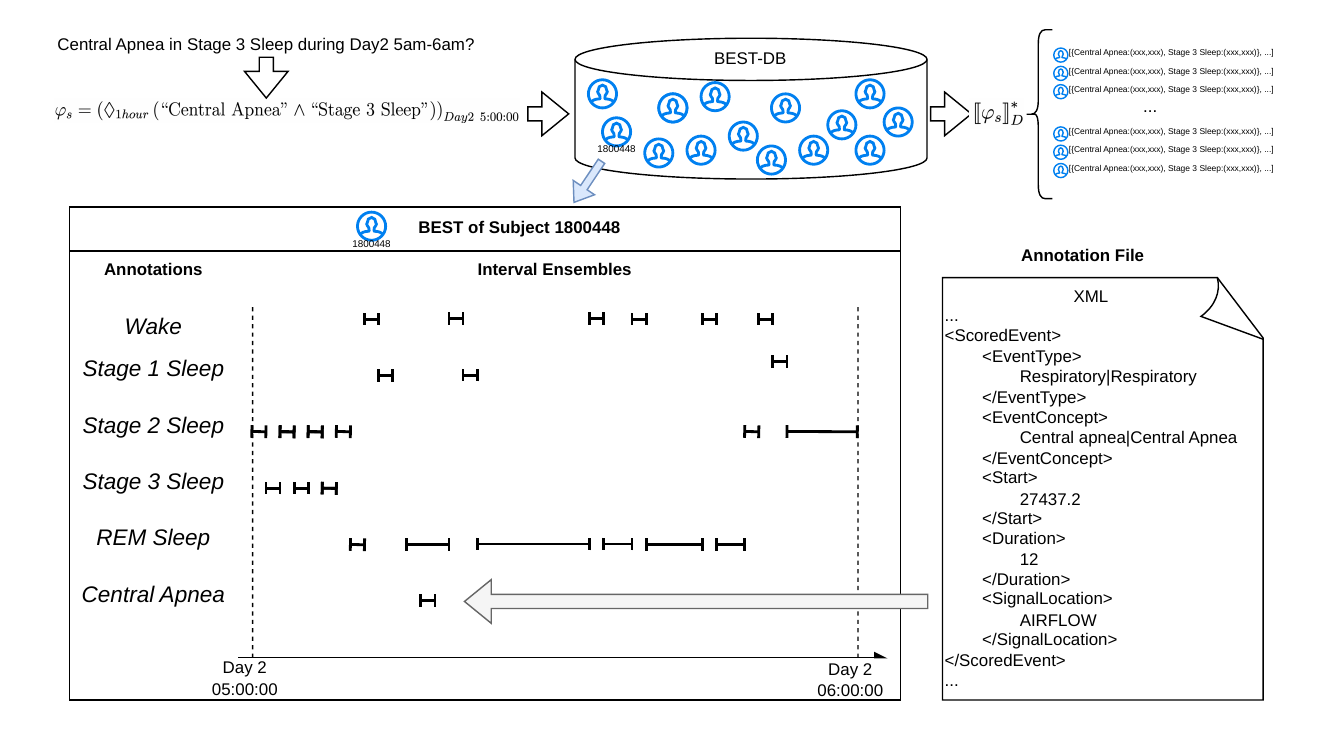}
  \caption{Overview of the BEST-DB framework. A ``Central Apnea'' annotation from NSRR XML is transformed into a labeled interval, inserted into the subject-specific BEST, and evaluated by QEL model checking. Query execution returns matching subjects and, when requested, the witnessing intervals that certify the temporal pattern.}
  \label{fig:bestdb_overview}
\end{figure}

\subsection{BEST for QEL over rational time}
\label{subsec:best-qel}
\textbf{Intervals and interval ensembles.} Following the QEL~\cite{qel} framework, the time domain is the rational line. For sleep data we use the non-negative rational timeline $\Qge$ because PSG annotation onsets and durations are stored as finite decimal or sampled quantities on a timeline starting at zero. An interval is denoted by $[\!(s,t)\!]$ and defined as
\[
[\!(s,t)\!] = \{z\in\Qge \mid s\preceq z \ll t\},
\]
where $s,t\in\Qge$ with $s\le t$, and the boundary relations $\preceq$ and $\ll$ are instantiated by $<$ or $\le$ to represent open, closed, left-closed right-open, and left-open right-closed intervals. A single time point is represented by the degenerate closed interval $[s,s]$.

\begin{definition}[Interval ensemble]
\label{def:interval_ensemble}
An \emph{interval ensemble} $\delta$ is a finite set of non-empty, pairwise non-overlapping intervals over $\Qge$. We write $\IE$ for the set of all interval ensembles over $\Qge$. In a normalized PSG representation, intervals in each ensemble are ordered by start time, and points not contained in the union of the intervals are treated under a closed-world assumption as times at which the corresponding label is absent.
\end{definition}

\textbf{Labeled interval ensembles.} Let $\Pset$ be the finite set of PSG event labels used as atomic propositions, such as $\Pset=\{\texttt{N3}, \texttt{REM}, \texttt{HYPOP}, \texttt{SPO2\_DESAT}, \texttt{ARASDA\_C3}, \texttt{CA}, and \texttt{OA}\}$. A labeled interval ensemble $(\ell,\delta)$ pairs $\ell\in\Pset$ with all intervals during which the label holds for a given subject or recording. For example, $(\texttt{N3},\delta_{n3})$ collects the Stage-3 sleep intervals in one timeline.

\begin{definition}[BEST and BEST-DB~\cite{qel}]
\label{def:BEST}
A \emph{Biomedical Event Structure Temporal Model} (BEST) is a function
\[
  \varepsilon:\Pset\to\IE.
\]
Thus $\varepsilon(\ell)$ is the interval ensemble in which label $\ell$ holds. A pair $(\ell,\varepsilon(\ell))$ is a labeled interval ensemble representing a temporal event of $\ell$. By convention, labels mapped to the empty interval ensemble are omitted from the displayed graph of $\varepsilon$, so a BEST may be written as a finite set of non-empty labeled interval ensembles. A BEST-DB $D$ is a finite collection of subject-level or session-level BESTs:
\[
D=\{\varepsilon^1,\varepsilon^2,\ldots,\varepsilon^m\}.
\]
\end{definition}

\subsection{QEL syntax and BEST semantics}
\label{subsec:qel-semantics}
We now restate the QEL syntax used by the application. Point terms over $\Qge$ are generated by constants, variables, and finite addition:
\[
 u,v ::= a \mid x \mid u+v, \qquad a\in\Qge.
\]
Window or norm terms are positive rational-valued terms generated by positive constants, variables, absolute values of point terms, and finite addition:
\[
 s,t ::= c \mid z \mid \lVert u\rVert \mid s+t, \qquad c\in\Qpos.
\]
For the sleep-query templates below, most window terms are constants or quantified duration variables; constrained quantifier notation such as $\exists x\in[0,r)$ is used as a readable abbreviation for restricting the assignment of $x$ to that rational interval.

\begin{definition}[QEL formula syntax~\cite{qel}]
\label{def:qel_syntax}
QEL formulas are generated by
\[
\varphi,\psi ::= p \mid \varphi_u \mid \neg\varphi \mid \varphi\wedge\psi \mid \varphi\vee\psi \mid \boldBox_t\varphi \mid \mdiamond_t\varphi \mid \exists x\,\varphi \mid \forall x\,\varphi,
\]
where $p\in\Pset$ is an atomic event label, $\varphi_u$ is exact displacement by point term $u$, $\mdiamond_t\varphi$ is bounded existence within the future window of length $t$, $\boldBox_t\varphi$ is bounded universality throughout that window, and quantified variables range over rational time or positive rational window lengths as appropriate.
\end{definition}

\begin{definition}[BEST satisfaction for QEL~\cite{qel}]
\label{def:qel_best_semantics}
Let $\varepsilon$ be a BEST and $s\in\Qge$ be a time point. The satisfaction relation $(\varepsilon,s)\models\varphi$ is defined inductively by the Boolean clauses of first-order logic and by the following QEL clauses:
\[
\begin{array}{ll}
(\varepsilon,s)\models p &\text{iff } s\in \bigcup\varepsilon(p),\\[2mm]
(\varepsilon,s)\models \varphi_u &\text{iff } (\varepsilon,s+u)\models\varphi,\\[2mm]
(\varepsilon,s)\models\mdiamond_t\varphi &\text{iff there exists } r\in\Qge \text{ with } 0\le r<t \text{ and } (\varepsilon,s+r)\models\varphi,\\[2mm]
(\varepsilon,s)\models\boldBox_t\varphi &\text{iff for all } r\in\Qge \text{ with } 0\le r<t,\ (\varepsilon,s+r)\models\varphi,\\[2mm]
(\varepsilon,s)\models\exists x\,\varphi &\text{iff there exists } a\in\Qge \text{ such that } (\varepsilon,s)\models\varphi[x\mapsto a],\\[2mm]
(\varepsilon,s)\models\forall x\,\varphi &\text{iff for all } a\in\Qge,\ (\varepsilon,s)\models\varphi[x\mapsto a].
\end{array}
\]
For positive duration variables $z$, the existential and universal clauses range over $\Qpos$. This is the future-directed QEL semantics specialized to PSG records.
\end{definition}

\subsection{Model checking and temporal query answers}
\label{subsec:model-checking-query}
A temporal query is a QEL formula evaluated against every BEST in a database. The answer set is
\[
\semantics{\varphi_q}_D = \{\varepsilon\in D\mid (\varepsilon,q)\models\varphi\}.
\]
Equivalently, using the QEL displacement convention, $\varepsilon\models\varphi_q$ abbreviates satisfaction of $\varphi$ at observation point $q$.
A standard QEL query formula $\varphi$ with an observation point $q$ follows the structure:
\[
\varphi_q \;:=\; \mathcal{P} \,(\psi_{t_1} \land \psi_{t_2} \land \dots \land \psi_{t_k})_q
\]
where:
\begin{enumerate}
    \item \textbf{Prefix Part $\mathcal{P}$:} This optional part consists of quantifiers ($\forall, \exists$) that declare time-distance variables used in the formula. If no variables are needed, this part is empty. The prefix scopes over the entire subsequent formula.
    
    \item \textbf{Core Part $\psi_{t_i}$:} The core consists of one or more sub-formulas connected by logical AND ($\land$). Each sub-formula $\psi_{t_i}$ is constructed from a set of atomic event annotations (labels) using one of the following operators:
      \[
      A_t=\{x\mid x+t\in A\},\qquad
      \mdiamond_r A=\{x\mid (x+[0,r))\cap A\ne\emptyset\},\qquad
      \boldBox_r A=\{x\mid x+[0,r)\subseteq A\}.
      \]
    The subscript $t_i$ denotes the temporal shift of the $i$-th sub-formula, meaning the condition is evaluated at time $q+t_i$. This allows correlating events at different relative times.
    Negation($\neg$) can be added to the atomic event annotations within $\psi_{t_i}$, and they are connected using logical AND ($\land$) and OR ($\lor$) as needed.
\end{enumerate}

\paragraph{Example.}
The query $\forall x{\in}\{1,2,3\}\,\exists y{\in}[0,5]\,\bigl((\varphi\land\varphi')_x \land (\Box_y(\psi\lor\neg\psi'))_3 \land (\Diamond_5(\neg\chi\land\chi'))_7\bigr)_{10}$
can be interpreted as: ``For all $x\in\{1,2,3\}$, there exists $y\in[0,5]$ such that at time $10{+}x$ both $\varphi$ and $\varphi'$ occur; at time $13$ event $\psi$ holds and $\psi'$ is absent for $y$ units; and at time $17$ within a 5-unit window $\chi$ is absent at some point while $\chi'$ occurs.''

\begin{definition}[Temporal pattern witness]
\label{def:temporal_pattern}
Let $\varepsilon\models\varphi_q$. A temporal pattern witness $\pi$ is a finite, set-minimal collection of positive labeled intervals and, when negation is used, negative gap intervals from complements of labeled ensembles, sufficient to certify the truth of $\varphi_q$ in $\varepsilon$. The query result may return only the matching BESTs, one witness per BEST, or all witnesses, denoted respectively by $\semantics{\varphi_q}_D$, $\semantics{\varphi_q}^{1}_D$, and $\semantics{\varphi_q}^{*}_D$.
\end{definition}

\subsection{Indexing for Labeled Interval Ensembles}
\label{sec:indexing}
The QEL semantics is defined over dense rational time, but each BEST contains only finitely many annotation intervals. This finiteness allows model checking to be reduced to finite evidence sets. The general QEL result uses quantifier elimination for dense ordered divisible abelian groups: after fixing a BEST, a reference point, and a formula, the satisfying assignments of rational variables decompose into finitely many linear cells, and it is sufficient to test one rational representative from each cell.

\begin{theorem}[Finite rational evidence set for QEL~\cite{qel}]
\label{prop:finite-boundary}
Fix a BEST $\varepsilon$, a reference point $q\in\Qge$, and a QEL formula $\varphi[\mathbf{x}]$ with free rational variables $\mathbf{x}=(x_1,\ldots,x_n)$. There exists a finite set $F_{\varepsilon,q,\varphi}\subseteq(\Qge)^n$, computable from the interval endpoints in $\varepsilon$, the rational constants in $\varphi$, the reference point $q$, and the affine boundaries induced by displacements and modal windows, such that
\[
(\varepsilon,q)\models\exists\mathbf{x}\,\varphi[\mathbf{x}]
\iff
\exists\mathbf{a}\in F_{\varepsilon,q,\varphi}\; (\varepsilon,q)\models\varphi[\mathbf{a}],
\]
\[
(\varepsilon,q)\models\forall\mathbf{x}\,\varphi[\mathbf{x}]
\iff
\forall\mathbf{a}\in F_{\varepsilon,q,\varphi}\; (\varepsilon,q)\models\varphi[\mathbf{a}].
\]
For the one-dimensional window queries used by the application, this finite set specializes to representatives of endpoint-induced segments and their bounding faces.
\end{theorem}
\begin{proof}[See~\cite{qel}] \end{proof}

For a fixed BEST, every atomic proposition is a finite union of rational intervals. QEL uses only addition, order, rational constants, displacement, and future-window existential or universal quantification. These clauses can be translated into the first-order theory of dense ordered divisible abelian groups, which admits quantifier elimination. The resulting quantifier-free formula is a Boolean combination of linear inequalities with rational constants, producing a finite cell decomposition. Truth is constant on each cell and face, so one rational representative per cell is complete for existential and universal model checking.

In the implementation, endpoint representatives are encoded as integer segment identifiers whenever the query contains only the bounded temporal templates in Section~\ref{sec:template}. Thus the finite-evidence theorem justifies the engineering reduction from dense rational time to endpoint arrays; the reduction is a semantic consequence of QEL, not merely an optimization.

\subsection{Temporal Query Template}
\label{sec:template}
We present standard QEL query templates for cohort extraction and pattern discovery. Let $\ell$ denote an event label, $\ell^*$ a wildcard label ranging over the event vocabulary, $q\in\Qge$ an observation point, $r\in\Qpos$ a duration, and $m$ the maximum duration of the recording. Cohort selection queries return matching BESTs; pattern extraction queries additionally return witnesses.

\paragraph{Single-event queries.}
These queries support direct data discovery and reusable feature extraction.
\begin{itemize}
    \item \textbf{Anywhere on the timeline:} $(\Diamond_m \ell)_0$.
    \item \textbf{Within $[q,q+r)$:} $(\Diamond_r \ell)_q$.
    \item \textbf{Starting within $[q,q+r)$:} $\exists x\in[0,r)\,\bigl((\Box_x\neg\ell)\land \ell_x\bigr)_q$.
\end{itemize}

\paragraph{Dual-event queries.}
These queries support phenotype definitions that depend on co-occurrence, precedence, and bounded delay between labels $\ell_A$ and $\ell_B$.
\begin{itemize}
  \item \textbf{Co-occurrence:} event $\ell_A$ overlaps event $\ell_B$ somewhere on the timeline:
  \[
    \bigl(\Diamond_m(\ell_A\land\ell_B)\bigr)_0.
  \]
  \item \textbf{Before variant 1} (entire A before B): 
  \[
    \exists q \exists x \exists y  
    \Bigl(
      \Box_x \neg\ell_A
      \;\land\;
      \bigl(\Box_y (\ell_A \land \neg\ell_B)\bigr)_{x}
      \;\land\;
      (\neg\ell_A)_{x+y}
      \;\land\;
      \bigl(\Diamond_m \ell_B\bigr)_{x+y}
    \Bigr)_q
  \]

    \item \textbf{Before variant 2} (A starts before B): 
    \[
      \exists q \exists x \exists y 
      \Bigl(
        \Box_x \neg\ell_A
        \;\land\;
        \bigl(\Box_y (\ell_A \land \neg\ell_B)\bigr)_{x}
        \;\land\;
      (\Diamond_m \ell_B)_{x+y}
    \Bigr)_q
    \]

    \item \textbf{Before variant 3} (A ends before B): 
    \[
      \exists q \exists x 
      \Bigl(
        \Box_x (\ell_A \land \neg\ell_B)
        \;\land\;
        (\neg\ell_A)_{x}
        \;\land\;
      (\Diamond_m \ell_B)_{x}
    \Bigr)_q
    \]

    \item \textbf{Before variant 4} (moment of A before B): 
    \[
      \exists q \exists x 
      \Bigl(
      \Box_x \left(\ell_A \land \neg\ell_B \right) \land \left(\Diamond_m \ell_B\right)_{x} 
      \Bigr)_q
    \]
\end{itemize}
The variants deliberately separate stricter interval precedence from weaker pointwise precedence, making the clinical meaning of a query explicit.

\paragraph{Event data extraction.}
These templates support data-management tasks in which matching event instances, not only matching subjects, are returned.
\begin{itemize}
  \item \textbf{Clock-anchored existence:} any label $\ell^*$ in $[q,q+r)$: $(\Diamond_r\ell^*)_q$.
  \item \textbf{Event-anchored existence:} $\ell$ co-occurs with any label $\ell^*$: $\bigl(\Diamond_m(\ell\land\ell^*)\bigr)_0$.
  \item \textbf{Event-anchored universality:} $\ell$ co-occurs with $\ell^*$ throughout a required window: $\bigl(\Box_m(\ell\land\ell^*)\bigr)_0$.
\end{itemize}

\subsection{Algorithms for Efficient Query Execution}
\label{sec:methods}

We deploy two structures under the QEL model-checking layer: 2DFC for single-event and extraction queries, where the task is fast interval retrieval within a window, and FCFC for dual-event pattern queries, where the task is fast temporal-distance lookup across heterogeneous event timelines.

\subsubsection{From Temporal Queries to Range Queries}

Because each interval $[x,y]$ can be treated as a point $(x,y)$ in 2D, a QEL existence query over a window $[x',y']$ reduces to a 2D range query: find all intervals whose start is not after the query end and whose end is not before the query start, i.e., all $(x,y)$ with $x\leq y'$ and $y\geq x'$. The baseline structure 2DRT~\cite{bentley_searching_problems,lueker_range_query} answers this in $O(\log^2 n+k)$ with $O(n\log n)$ space; RTFC~\cite{rtfc} improves query time to $O(\log n+k)$ via fractional cascading~\cite{fc} at the same space cost. Because BEST intervals are pre-sorted and non-overlapping, 2DFC eliminates the preprocessing overhead of these general structures.

\subsubsection{2-Dimensional Fractional Cascading for Intervals (2DFC)}
\label{2dfc}
Given $m$ sorted non-overlapping interval lists with $n$ total intervals, 2DFC splits each list into a start-point list $L_\ell^x$ and an end-point list $L_\ell^y$, then builds two 1DFC structures ($FC_x$, $FC_y$). Algorithm~\ref{alg:build2DFC} details the construction.

\begin{algorithm}
  \caption{Build 2DFC}
  \label{alg:build2DFC}
  \SetKwFunction{Build}{Build2DFC}
  \KwIn{The original 2D array of intervals $L$}
  \KwOut{2 1DFC data structures $FC_x$ and $FC_y$}
  \SetKwProg{Fn}{Function}{:}{}
  \Fn{\Build{$L$}}{
    Initialize $L_x$ and $L_y$\;
    \For{$i \leftarrow 0$ \KwTo $|L|-1$}{
      \ForEach{$interval$ \textbf{in} $L[i]$}{
        Add $interval.start$ to $L_x[i]$\;
        Add $interval.end$ to $L_y[i]$\;
      }
    }
    Create 1DFC $FC_x$ for $L_x$ using the structure of $(data,index\_in\_L,pointer\_next)$\;
    Create 1DFC $FC_y$ for $L_y$ using the structure of $(data,index\_in\_L,pointer\_next)$\;
    \KwRet{$FC_x$, $FC_y$}\;
  } 
\end{algorithm}

\begin{algorithm}
  \caption{Existence query by 2DFC}
  \label{alg:query2DFC}
  \SetKwFunction{QODFC}{Query1DFC}
  \SetKwFunction{QTDFC}{Query2DFC}
  \KwIn{$L$, $FC_x$, $FC_y$, $q=[x',y']$}
  \KwOut{A list of intervals}
  \SetKwProg{Fn}{Function}{:}{}
  \Fn{\QTDFC{$L$, $FC_x$, $FC_y$, $x'$, $y'$}}{
    Initialize $Result$\;
    $start\_indices \leftarrow \QODFC(FC_y, x')$\;
    $end\_indices \leftarrow \QODFC(FC_x, y')$\;
    \For{$i \leftarrow 0$ \KwTo $|L|-1$}{
      \For{$j \leftarrow start\_indices[i]$ \KwTo $end\_indices[i]$}{
        \If{$L[i][j].start \leq y'$}{
          Add $L[i][j]$ to $Result$\;
        }
      }
    }
    \KwRet{$Result$}\;
  } 

  \Fn{\QODFC{$L$, $FC$, $x$}}{
    Initialize $Result$\;
    $pointer \leftarrow \text{BinarySearch}(FC[0], x)$\;
    \For{$i \leftarrow 0$ \KwTo $|FC|-1$}{
      \If{$pointer < FC[i].length $}{
        $index\_in\_L \leftarrow FC[i][pointer][1]$\;
        $pointer \leftarrow FC[i][pointer][2]$\;
      }
      \Else{
        $index\_in\_L \leftarrow L[i].length$\;
        $pointer \leftarrow FC[i+1].length$\;
      }
      \If{$pointer > 0$ and $FC[i+1][pointer - 1][0] \leq x$}{
        $pointer \leftarrow pointer - 1$\;
      }
      Add $index\_in\_L$ to $Result$\;
    }
    \KwRet{$Result$}\;
  }
\end{algorithm}

\begin{proposition}\label{prop:2dfc_correctness}
Given a valid 2DFC structure over a set of non-overlapping interval ensembles, Algorithm~\ref{alg:query2DFC} returns exactly the set of intervals $[x_i,y_i]$ that overlap the query interval $[x',y']$, equivalently those satisfying $x_i\leq y'$ and $y_i\geq x'$. Thus 2DFC correctly solves the interval-overlap range query required by QEL bounded-existence extraction.
\end{proposition}

The QEL query algorithm for an existence query using 2DFC is shown in Algorithm~\ref{alg:query2DFC}; Figure~\ref{2dfc_example} illustrates the cascading pointer traversal on a worked example.

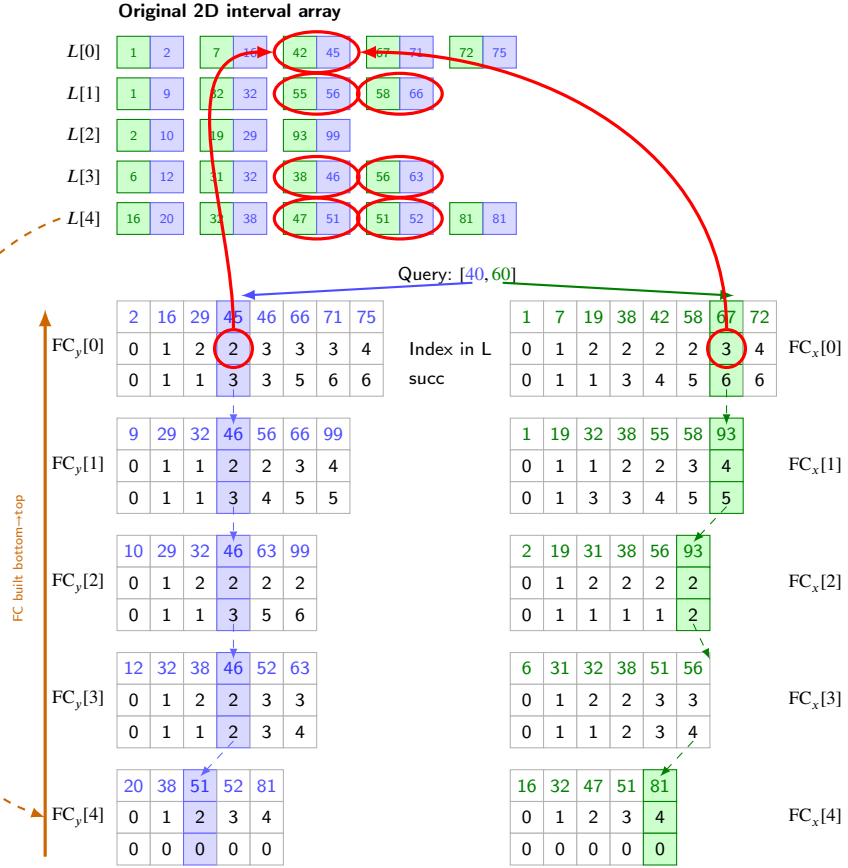
\begin{figure}
  \centering
  \begin{tikzpicture}[
    font=\scriptsize,
    >=latex,
    NC/.style={draw=gray!60,fill=white,minimum width=0.44cm,minimum height=0.42cm,inner sep=0.6pt,anchor=center},
    GC/.style={draw=green!50!black,fill=green!20,minimum width=0.44cm,minimum height=0.42cm,inner sep=0.6pt,anchor=center},
    BC/.style={draw=blue!60,fill=blue!15,minimum width=0.44cm,minimum height=0.42cm,inner sep=0.6pt,anchor=center},
    RC/.style={draw=red!70,fill=white,minimum width=0.44cm,minimum height=0.42cm,inner sep=0.6pt,anchor=center,very thick},
  ]

  \def\CW{0.44}
  \def\RH{0.42}
  \def\TR{0.42}
  \def\LX{0.0}
  \def\RX{5.2}
  \def\TYa{0.00}
  \def\TYb{1.55}
  \def\TYc{3.10}
  \def\TYd{4.65}
  \def\TYe{6.20}

  \def\ILY{3.5}
  \def\ILS{0.55}

  \node[anchor=south west, font=\scriptsize\bfseries] at (-0.1, \ILY+0.30)
    {Original 2D interval array};

  \foreach \ii/\lbl in {0/$L[0]$,1/$L[1]$,2/$L[2]$,3/$L[3]$,4/$L[4]$}{
    \node[anchor=east,font=\scriptsize] at (-0.1, \ILY-\ii*\ILS) {\lbl};
  }

  \foreach \kk/\es/\ee in {0/1/2, 1/7/16, 2/42/45, 3/67/71, 4/72/75}{
    \pgfmathsetmacro{\ix}{\kk*1.1}
    \node[GC,font=\tiny] at (\ix+0.22, \ILY) {\textcolor{green!50!black}{\es}};
    \node[BC,font=\tiny] at (\ix+0.66, \ILY) {\textcolor{blue!70}{\ee}};
  }
  \foreach \kk/\es/\ee in {0/1/9, 1/32/32, 2/55/56, 3/58/66}{
    \pgfmathsetmacro{\ix}{\kk*1.1}
    \node[GC,font=\tiny] at (\ix+0.22, \ILY-1*\ILS) {\textcolor{green!50!black}{\es}};
    \node[BC,font=\tiny] at (\ix+0.66, \ILY-1*\ILS) {\textcolor{blue!70}{\ee}};
  }
  \foreach \kk/\es/\ee in {0/2/10, 1/19/29, 2/93/99}{
    \pgfmathsetmacro{\ix}{\kk*1.1}
    \node[GC,font=\tiny] at (\ix+0.22, \ILY-2*\ILS) {\textcolor{green!50!black}{\es}};
    \node[BC,font=\tiny] at (\ix+0.66, \ILY-2*\ILS) {\textcolor{blue!70}{\ee}};
  }
  \foreach \kk/\es/\ee in {0/6/12, 1/31/32, 2/38/46, 3/56/63}{
    \pgfmathsetmacro{\ix}{\kk*1.1}
    \node[GC,font=\tiny] at (\ix+0.22, \ILY-3*\ILS) {\textcolor{green!50!black}{\es}};
    \node[BC,font=\tiny] at (\ix+0.66, \ILY-3*\ILS) {\textcolor{blue!70}{\ee}};
  }
  \foreach \kk/\es/\ee in {0/16/20, 1/32/38, 2/47/51, 3/51/52, 4/81/81}{
    \pgfmathsetmacro{\ix}{\kk*1.1}
    \node[GC,font=\tiny] at (\ix+0.22, \ILY-4*\ILS) {\textcolor{green!50!black}{\es}};
    \node[BC,font=\tiny] at (\ix+0.66, \ILY-4*\ILS) {\textcolor{blue!70}{\ee}};
  }


  \foreach \ci/\vv/\hi in {%
      0/2/0, 1/16/0, 2/29/0, 3/45/1, 4/46/0, 5/66/0, 6/71/0, 7/75/0}{
    \pgfmathsetmacro{\cx}{\LX+\ci*\CW+\CW/2}
    \ifnum\hi=1 \node[BC] at (\cx, -\TYa) {\textcolor{blue!80}{\vv}};
    \else       \node[NC] at (\cx, -\TYa) {\textcolor{blue!70}{\vv}}; \fi
  }
  \foreach \ci/\vv/\hi in {%
      0/0/0, 1/1/0, 2/2/0, 3/2/1, 4/3/0, 5/3/0, 6/3/0, 7/4/0}{
    \pgfmathsetmacro{\cx}{\LX+\ci*\CW+\CW/2}
    \ifnum\hi=1 \node[BC] at (\cx, -\TYa-\RH) {\vv};
    \else       \node[NC] at (\cx, -\TYa-\RH) {\vv}; \fi
  }
  \foreach \ci/\vv/\hi in {%
      0/0/0, 1/1/0, 2/1/0, 3/3/1, 4/3/0, 5/5/0, 6/6/0, 7/6/0}{
    \pgfmathsetmacro{\cx}{\LX+\ci*\CW+\CW/2}
    \ifnum\hi=1 \node[BC] at (\cx, -\TYa-2*\RH) {\vv};
    \else       \node[NC] at (\cx, -\TYa-2*\RH) {\vv}; \fi
  }
  \pgfmathsetmacro{\rcx}{\LX+3*\CW+\CW/2}   
  \draw[red,very thick] (\rcx, -\TYa-\RH) circle (0.25cm);

  \foreach \ci/\vv/\hi in {%
      0/9/0, 1/29/0, 2/32/0, 3/46/1, 4/56/0, 5/66/0, 6/99/0}{
    \pgfmathsetmacro{\cx}{\LX+\ci*\CW+\CW/2}
    \ifnum\hi=1 \node[BC] at (\cx,-\TYb) {\textcolor{blue!80}{\vv}};
    \else       \node[NC] at (\cx,-\TYb) {\textcolor{blue!70}{\vv}}; \fi
  }
  \foreach \ci/\vv/\hi in {%
      0/0/0, 1/1/0, 2/1/0, 3/2/1, 4/2/0, 5/3/0, 6/4/0}{
    \pgfmathsetmacro{\cx}{\LX+\ci*\CW+\CW/2}
    \ifnum\hi=1 \node[BC] at (\cx,-\TYb-\RH) {\vv};
    \else       \node[NC] at (\cx,-\TYb-\RH) {\vv}; \fi
  }
  \foreach \ci/\vv/\hi in {%
      0/0/0, 1/1/0, 2/1/0, 3/3/1, 4/4/0, 5/5/0, 6/5/0}{
    \pgfmathsetmacro{\cx}{\LX+\ci*\CW+\CW/2}
    \ifnum\hi=1 \node[BC] at (\cx,-\TYb-2*\RH) {\vv};
    \else       \node[NC] at (\cx,-\TYb-2*\RH) {\vv}; \fi
  }

  \foreach \ci/\vv/\hi in {%
      0/10/0, 1/29/0, 2/32/0, 3/46/1, 4/63/0, 5/99/0}{
    \pgfmathsetmacro{\cx}{\LX+\ci*\CW+\CW/2}
    \ifnum\hi=1 \node[BC] at (\cx,-\TYc) {\textcolor{blue!80}{\vv}};
    \else       \node[NC] at (\cx,-\TYc) {\textcolor{blue!70}{\vv}}; \fi
  }
  \foreach \ci/\vv/\hi in {%
      0/0/0, 1/1/0, 2/2/0, 3/2/1, 4/2/0, 5/2/0}{
    \pgfmathsetmacro{\cx}{\LX+\ci*\CW+\CW/2}
    \ifnum\hi=1 \node[BC] at (\cx,-\TYc-\RH) {\vv};
    \else       \node[NC] at (\cx,-\TYc-\RH) {\vv}; \fi
  }
  \foreach \ci/\vv/\hi in {%
      0/0/0, 1/1/0, 2/1/0, 3/3/1, 4/5/0, 5/6/0}{
    \pgfmathsetmacro{\cx}{\LX+\ci*\CW+\CW/2}
    \ifnum\hi=1 \node[BC] at (\cx,-\TYc-2*\RH) {\vv};
    \else       \node[NC] at (\cx,-\TYc-2*\RH) {\vv}; \fi
  }

  \foreach \ci/\vv/\hi in {%
      0/12/0, 1/32/0, 2/38/0, 3/46/1, 4/52/0, 5/63/0}{
    \pgfmathsetmacro{\cx}{\LX+\ci*\CW+\CW/2}
    \ifnum\hi=1 \node[BC] at (\cx,-\TYd) {\textcolor{blue!80}{\vv}};
    \else       \node[NC] at (\cx,-\TYd) {\textcolor{blue!70}{\vv}}; \fi
  }
  \foreach \ci/\vv/\hi in {%
      0/0/0, 1/1/0, 2/2/0, 3/2/1, 4/3/0, 5/3/0}{
    \pgfmathsetmacro{\cx}{\LX+\ci*\CW+\CW/2}
    \ifnum\hi=1 \node[BC] at (\cx,-\TYd-\RH) {\vv};
    \else       \node[NC] at (\cx,-\TYd-\RH) {\vv}; \fi
  }
  \foreach \ci/\vv/\hi in {%
      0/0/0, 1/1/0, 2/1/0, 3/2/1, 4/3/0, 5/4/0}{
    \pgfmathsetmacro{\cx}{\LX+\ci*\CW+\CW/2}
    \ifnum\hi=1 \node[BC] at (\cx,-\TYd-2*\RH) {\vv};
    \else       \node[NC] at (\cx,-\TYd-2*\RH) {\vv}; \fi
  }

  \foreach \ci/\vv/\hi in {%
      0/20/0, 1/38/0, 2/51/1, 3/52/0, 4/81/0}{
    \pgfmathsetmacro{\cx}{\LX+\ci*\CW+\CW/2}
    \ifnum\hi=1 \node[BC] at (\cx,-\TYe) {\textcolor{blue!80}{\vv}};
    \else       \node[NC] at (\cx,-\TYe) {\textcolor{blue!70}{\vv}}; \fi
  }
  \foreach \ci/\vv/\hi in {%
      0/0/0, 1/1/0, 2/2/1, 3/3/0, 4/4/0}{
    \pgfmathsetmacro{\cx}{\LX+\ci*\CW+\CW/2}
    \ifnum\hi=1 \node[BC] at (\cx,-\TYe-\RH) {\vv};
    \else       \node[NC] at (\cx,-\TYe-\RH) {\vv}; \fi
  }
  \foreach \ci/\vv/\hi in {%
      0/0/0, 1/0/0, 2/0/1, 3/0/0, 4/0/0}{
    \pgfmathsetmacro{\cx}{\LX+\ci*\CW+\CW/2}
    \ifnum\hi=1 \node[BC] at (\cx,-\TYe-2*\RH) {\vv};
    \else       \node[NC] at (\cx,-\TYe-2*\RH) {\vv}; \fi
  }


  \foreach \ci/\vv/\hi in {%
      0/1/0, 1/7/0, 2/19/0, 3/38/0, 4/42/0, 5/58/0, 6/67/1, 7/72/0}{
    \pgfmathsetmacro{\cx}{\RX+\ci*\CW+\CW/2}
    \ifnum\hi=1 \node[GC] at (\cx,-\TYa) {\textcolor{green!50!black}{\vv}};
    \else       \node[NC] at (\cx,-\TYa) {\textcolor{green!50!black}{\vv}}; \fi
  }
  \foreach \ci/\vv/\hi in {%
      0/0/0, 1/1/0, 2/2/0, 3/2/0, 4/2/0, 5/2/0, 6/3/1, 7/4/0}{
    \pgfmathsetmacro{\cx}{\RX+\ci*\CW+\CW/2}
    \ifnum\hi=1 \node[GC] at (\cx,-\TYa-\RH) {\vv};
    \else       \node[NC] at (\cx,-\TYa-\RH) {\vv}; \fi
  }
  \foreach \ci/\vv/\hi in {%
      0/0/0, 1/1/0, 2/1/0, 3/3/0, 4/4/0, 5/5/0, 6/6/1, 7/6/0}{
    \pgfmathsetmacro{\cx}{\RX+\ci*\CW+\CW/2}
    \ifnum\hi=1 \node[GC] at (\cx,-\TYa-2*\RH) {\vv};
    \else       \node[NC] at (\cx,-\TYa-2*\RH) {\vv}; \fi
  }
  \pgfmathsetmacro{\rrcx}{\RX+6*\CW+\CW/2}   
  \draw[red,very thick] (\rrcx, -\TYa-\RH) circle (0.25cm);

  \foreach \ci/\vv/\hi in {%
      0/1/0, 1/19/0, 2/32/0, 3/38/0, 4/55/0, 5/58/0, 6/93/1}{
    \pgfmathsetmacro{\cx}{\RX+\ci*\CW+\CW/2}
    \ifnum\hi=1 \node[GC] at (\cx,-\TYb) {\textcolor{green!50!black}{\vv}};
    \else       \node[NC] at (\cx,-\TYb) {\textcolor{green!50!black}{\vv}}; \fi
  }
  \foreach \ci/\vv/\hi in {%
      0/0/0, 1/1/0, 2/1/0, 3/2/0, 4/2/0, 5/3/0, 6/4/1}{
    \pgfmathsetmacro{\cx}{\RX+\ci*\CW+\CW/2}
    \ifnum\hi=1 \node[GC] at (\cx,-\TYb-\RH) {\vv};
    \else       \node[NC] at (\cx,-\TYb-\RH) {\vv}; \fi
  }
  \foreach \ci/\vv/\hi in {%
      0/0/0, 1/1/0, 2/3/0, 3/3/0, 4/4/0, 5/5/0, 6/5/1}{
    \pgfmathsetmacro{\cx}{\RX+\ci*\CW+\CW/2}
    \ifnum\hi=1 \node[GC] at (\cx,-\TYb-2*\RH) {\vv};
    \else       \node[NC] at (\cx,-\TYb-2*\RH) {\vv}; \fi
  }

  \foreach \ci/\vv/\hi in {%
      0/2/0, 1/19/0, 2/31/0, 3/38/0, 4/56/0, 5/93/1}{
    \pgfmathsetmacro{\cx}{\RX+\ci*\CW+\CW/2}
    \ifnum\hi=1 \node[GC] at (\cx,-\TYc) {\textcolor{green!50!black}{\vv}};
    \else       \node[NC] at (\cx,-\TYc) {\textcolor{green!50!black}{\vv}}; \fi
  }
  \foreach \ci/\vv/\hi in {%
      0/0/0, 1/1/0, 2/2/0, 3/2/0, 4/2/0, 5/2/1}{
    \pgfmathsetmacro{\cx}{\RX+\ci*\CW+\CW/2}
    \ifnum\hi=1 \node[GC] at (\cx,-\TYc-\RH) {\vv};
    \else       \node[NC] at (\cx,-\TYc-\RH) {\vv}; \fi
  }
  \foreach \ci/\vv/\hi in {%
      0/0/0, 1/1/0, 2/1/0, 3/1/0, 4/1/0, 5/2/1}{
    \pgfmathsetmacro{\cx}{\RX+\ci*\CW+\CW/2}
    \ifnum\hi=1 \node[GC] at (\cx,-\TYc-2*\RH) {\vv};
    \else       \node[NC] at (\cx,-\TYc-2*\RH) {\vv}; \fi
  }

  \foreach \ci/\vv/\hi in {%
      0/6/0, 1/31/0, 2/32/0, 3/38/0, 4/51/0, 5/56/0}{
    \pgfmathsetmacro{\cx}{\RX+\ci*\CW+\CW/2}
    \ifnum\hi=1 \node[GC] at (\cx,-\TYd) {\textcolor{green!50!black}{\vv}};
    \else       \node[NC] at (\cx,-\TYd) {\textcolor{green!50!black}{\vv}}; \fi
  }
  \foreach \ci/\vv/\hi in {%
      0/0/0, 1/1/0, 2/2/0, 3/2/0, 4/3/0, 5/3/0}{
    \pgfmathsetmacro{\cx}{\RX+\ci*\CW+\CW/2}
    \ifnum\hi=1 \node[GC] at (\cx,-\TYd-\RH) {\vv};
    \else       \node[NC] at (\cx,-\TYd-\RH) {\vv}; \fi
  }
  \foreach \ci/\vv/\hi in {%
      0/0/0, 1/1/0, 2/1/0, 3/2/0, 4/3/0, 5/4/0}{
    \pgfmathsetmacro{\cx}{\RX+\ci*\CW+\CW/2}
    \ifnum\hi=1 \node[GC] at (\cx,-\TYd-2*\RH) {\vv};
    \else       \node[NC] at (\cx,-\TYd-2*\RH) {\vv}; \fi
  }

  \foreach \ci/\vv/\hi in {%
      0/16/0, 1/32/0, 2/47/0, 3/51/0, 4/81/1}{
    \pgfmathsetmacro{\cx}{\RX+\ci*\CW+\CW/2}
    \ifnum\hi=1 \node[GC] at (\cx,-\TYe) {\textcolor{green!50!black}{\vv}};
    \else       \node[NC] at (\cx,-\TYe) {\textcolor{green!50!black}{\vv}}; \fi
  }
  \foreach \ci/\vv/\hi in {%
      0/0/0, 1/1/0, 2/2/0, 3/3/0, 4/4/1}{
    \pgfmathsetmacro{\cx}{\RX+\ci*\CW+\CW/2}
    \ifnum\hi=1 \node[GC] at (\cx,-\TYe-\RH) {\vv};
    \else       \node[NC] at (\cx,-\TYe-\RH) {\vv}; \fi
  }
  \foreach \ci/\vv/\hi in {%
      0/0/0, 1/0/0, 2/0/0, 3/0/0, 4/0/1}{
    \pgfmathsetmacro{\cx}{\RX+\ci*\CW+\CW/2}
    \ifnum\hi=1 \node[GC] at (\cx,-\TYe-2*\RH) {\vv};
    \else       \node[NC] at (\cx,-\TYe-2*\RH) {\vv}; \fi
  }

  \foreach \ii/\ty in {0/\TYa, 1/\TYb, 2/\TYc, 3/\TYd, 4/\TYe}{
    \node[anchor=east,font=\scriptsize] at (\LX-0.05, -\ty-\RH)
      {$\mathrm{FC}_y[\ii]$};
    \node[anchor=west,font=\scriptsize] at (\RX+8*\CW+0.05, -\ty-\RH)
      {$\mathrm{FC}_x[\ii]$};
  }

  \node[anchor=west,font=\scriptsize,align=left] at (\LX+3.75,-\TYa-\RH)   {Index in L};
  \node[anchor=west,font=\scriptsize,align=left] at (\LX+3.75,-\TYa-2*\RH) {succ};

  \node[font=\scriptsize] at (4.5, 0.55)
    {Query: $[\textcolor{blue!70}{40},\textcolor{green!50!black}{60}]$};

  \pgfmathsetmacro{\ax}{\LX+3*\CW+\CW/2}
  \draw[->,blue!60,dashed] (\ax,-\TYa-2*\RH-0.12) -- (\ax,-\TYb+0.12);
  \draw[->,blue!60,dashed] (\ax,-\TYb-2*\RH-0.12) -- (\ax,-\TYc+0.12);
  \draw[->,blue!60,dashed] (\ax,-\TYc-2*\RH-0.12) -- (\ax,-\TYd+0.12);
  \pgfmathsetmacro{\axe}{\LX+2*\CW+\CW/2}
  \draw[->,blue!60,dashed] (\ax,-\TYd-2*\RH-0.12) -- (\axe,-\TYe+0.12);

  \pgfmathsetmacro{\bx}{\RX+6*\CW+\CW/2}
  \draw[->,green!50!black,dashed] (\bx,-\TYa-2*\RH-0.12) -- (\bx,-\TYb+0.12);
  \pgfmathsetmacro{\bxb}{\RX+5*\CW+\CW/2}
  \draw[->,green!50!black,dashed] (\bx,-\TYb-2*\RH-0.12) -- (\bxb,-\TYc+0.12);
  \pgfmathsetmacro{\bxcend}{\RX+6*\CW}
  \pgfmathsetmacro{\bxcsrc}{\RX+5*\CW+\CW/2}
  \draw[->,green!50!black,dashed] (\bxb,-\TYc-2*\RH-0.12) -- (\bxcend,-\TYd+0.12);
  \pgfmathsetmacro{\bxd}{\RX+4*\CW+\CW/2}
  \draw[->,green!50!black,dashed] (\bxcsrc,-\TYd-2*\RH-0.12) -- (\bxd,-\TYe+0.12);

  \pgfmathsetmacro{\qrrcx}{\RX+6*\CW+\CW/2}   
  \pgfmathsetmacro{\qrcx}{\LX+3*\CW+\CW/2}    
  \draw[->,green!50!black,thick] (5.1,0.45) -- (\qrrcx+0.1,-\TYa+0.28);
  \draw[->,blue!70,thick]        (4.7,0.45) -- (\qrcx+0.1,-\TYa+0.28);

  \draw[->, red, very thick]
    (\rcx, -\TYa-\RH+0.27)
    to[out=90, in=180]
    (2.08, \ILY);

  \draw[->, red, very thick]
    (\rrcx, -\TYa-\RH+0.27)
    to[out=90, in=0]
    (3.20, \ILY);

  \draw[->,orange!80!black, thick, dashed]
    (-0.75, \ILY-4*\ILS)           
    to[out=200, in=160]
    (-0.95, -\TYe-\RH);            

  \draw[->, orange!80!black, very thick]
    (-0.95, -\TYe-2*\RH-0.10)     
    -- (-0.95, -\TYa+0.10);        

  \node[font=\tiny, orange!80!black, rotate=90, anchor=south] at (-1.1, -3.2)
    {FC built bottom$\to$top};

  \foreach \ecx/\ecy in {%
    2.64/3.50,%
    2.64/2.95, 3.74/2.95,%
    2.64/1.85, 3.74/1.85,%
    2.64/1.30, 3.74/1.30%
  }{
    \draw[red,very thick] (\ecx, \ecy) ellipse (0.56 and 0.27);
  }

  \end{tikzpicture}
  \caption{\small An example to illustrate the construction of 2DFC and the existence query
    with input interval $[40,60]$. Extracted target ranges are navigated via pointers (boxes),
    triggering a localized scan (arrows), and matching intervals are returned (circles).}
  \label{2dfc_example}
\end{figure}

\paragraph{Complexity analysis.}
Leveraging the FC data structure, given a 2-dimensional query interval $q=[x',y']$ and a dataset $L$ comprising $m$ sorted interval lists with a total of $n$ intervals, we can efficiently locate the position of $x'$ in every end-point list and the position of $y'$ in every start-point list. In Algorithm~\ref{alg:query2DFC}, the \textsc{Query1DFC} procedure performs a single binary search on the augmented first list $FC[0]$ and then cascades through the remaining $m-1$ lists in $O(1)$ time per list. The size of $FC[0]$ is bounded by $|FC[0]| \leq n_1 + \tfrac{n_2}{2} + \tfrac{n_3}{4} + \cdots$. Since $\sum_{i \geq 2} \frac{n_i}{2^{i-1}} \leq \sum_{i \geq 1} n_i = n$, the total is at most $n_1 + n \leq 2n$, where $n_i$ denotes the number of intervals in the $i$-th list. When the lists are approximately balanced (i.e., $n_i \approx n/m$), we have $|FC[0]| \leq 2n/m = 2n_1$, so the binary search takes $O(\log n_1)$ time. In the general (unbalanced) case, $|FC[0]|$ can be as large as $O(n)$, making the binary search $O(\log n)$.
Since the algorithm invokes \textsc{Query1DFC} twice (once for start points, once for end points) and then scans the candidate intervals in $O(m + k)$ time, the total query time is:
\[
O(\log |FC[0]| + m + k),
\]
which is $O(\log n_1 + m + k)$ when the lists are balanced and $O(\log n + m + k)$ in the worst case. Here $k$ is the number of reported intervals. Compared to the 2DRT query time of $O(\log^2 n + k)$ and the RTFC query time of $O(\log n + k)$, 2DFC trades a small additive $O(m)$ term for a potentially reduced logarithmic factor. In typical sleep study datasets where $m \ll n$ (e.g., $m = 20$ event types and $n > 200{,}000$ intervals), the $O(m)$ term is negligible and 2DFC achieves competitive or superior query performance.

\paragraph{Space complexity.}
By the standard fractional cascading result~\cite{fc}, the total number of elements across all lists in a single FC structure is at most $2n$. Since 2DFC constructs two independent FC structures (one for start points, one for end points), the total storage is at most $4n = O(n)$. This is a significant improvement over both 2DRT and RTFC, which require $O(n \log n)$ space due to the replication of elements across $O(\log n)$ levels of the range tree.

\subsubsection{Fully Connected Fractional Cascading for Pattern Matching (FCFC)}
\label{fcfc}
While 2DFC is designed for efficient event data extraction, such as retrieving all events that occur within a given time window, pattern matching queries require a complementary strategy. In pattern matching, the goal is to identify temporal relationships (e.g., co-occurrence, precedence, bounded delay) between a specified \emph{anchor} event and one or more \emph{target} events across the entire timeline. This corresponds directly to the dual-event query templates defined in Section~\ref{sec:template}, such as the ``Before'' variants and ``Co-occurrence'' patterns.

As discussed in Section~\ref{sec:queries}, NFA-based approaches (e.g., \texttt{MATCH\_RECOGNIZE}~\cite{iso2016sql}) require flattening concurrent interval streams into a single row sequence, lack first-class metric temporal operators, and do not exploit the non-overlapping BEST structure. To overcome these limitations, we introduce FCFC, which connects all interval boundaries to a global boundary index and precomputes event-specific successor arrays. The term ``fractional cascading'' in FCFC refers to this propagation of boundary positions through a fully connected global backbone, rather than the classical list-sampling structure used in 2DFC.

The FCFC structure is built from the sorted, non-overlapping interval ensembles in a BEST. Let $Global$ be the sorted set of all starts and ends across all event labels. For each label $i$ and each global boundary $Global[j]$, FCFC stores (i) the first interval of label $i$ whose start is not before $Global[j]$, supporting bounded-future and precedence tests, and (ii) the first interval of label $i$ whose end is after $Global[j]$, supporting overlap/co-occurrence tests. Because interval boundaries are also linked to their positions in $Global$ during construction, a query can move from an anchor interval boundary to the relevant target-event candidate in $O(1)$ time per target label.

\begin{algorithm}
  \caption{Build FCFC}
  \label{alg:buildFCFC}
  \SetKwFunction{BuildFCFC}{BuildFCFC}
  \KwIn{The original 2D array of sorted, non-overlapping intervals $L$}
  \KwOut{Global index array $Global$, next-start arrays $NextStart$, next-live arrays $NextLive$, and boundary links}
  \SetKwProg{Fn}{Function}{:}{}
  \Fn{\BuildFCFC{$L$}}{
    $Global \leftarrow \text{Sort unique starts and ends across all lists in } L$\;
    $NextStart \leftarrow \text{Empty array of size } |L| \times |Global|$\;
    $NextLive \leftarrow \text{Empty array of size } |L| \times |Global|$\;
    \For{$i \leftarrow 0$ \KwTo $|L|-1$}{
      $s \leftarrow 0$; $e \leftarrow 0$\;
      \For{$j \leftarrow 0$ \KwTo $|Global|-1$}{
        \While{$s < |L[i]|$ and $L[i][s].start < Global[j]$}{
          $s \leftarrow s + 1$\;
        }
        \While{$e < |L[i]|$ and $L[i][e].end \leq Global[j]$}{
          $e \leftarrow e + 1$\;
        }
        $NextStart[i][j] \leftarrow s$ \tcp*{first interval starting at or after $Global[j]$}
        $NextLive[i][j] \leftarrow e$ \tcp*{first interval ending after $Global[j]$}
      }
    }
    Link each interval boundary in $L$ to its index in $Global$\;
    \KwRet{$Global, NextStart, NextLive$}\;
  }
\end{algorithm}

\begin{algorithm}
  \caption{Pattern Matching Query (Co-occurrence) by FCFC}
  \label{alg:queryFCFC}
  \SetKwFunction{QueryFCFC}{QueryFCFC}
  \KwIn{Interval lists $L$, anchor event index $A$, target event index $T$, global-index links, and $NextLive$}
  \KwOut{Intervals of $A$ overlapping at least one interval of $T$}
  \SetKwProg{Fn}{Function}{:}{}
  \Fn{\QueryFCFC{$L$, $A$, $T$, $NextLive$}}{
    $Result \leftarrow \text{Empty list}$\;
    \ForEach{$interval$ \textbf{in} $L[A]$}{
      $j \leftarrow interval.global\_start\_idx$\;
      $c \leftarrow NextLive[T][j]$\;
      \If{$c < |L[T]|$ and $L[T][c].start < interval.end$}{
        Add $interval$ to $Result$\;
      }
    }
    \KwRet{$Result$}\;
  }
\end{algorithm}

\begin{proposition}
Given a valid FCFC structure over normalized non-overlapping interval ensembles, Algorithm~\ref{alg:queryFCFC} returns exactly the anchor intervals of label $A$ that overlap at least one target interval of label $T$.
\end{proposition}

\noindent\emph{Proof sketch.}
For an anchor interval $[a,b)$, the FCFC lookup $NextLive[T][j]$ at the global position of $a$ returns the first target interval whose end is greater than $a$. All earlier target intervals end at or before $a$ and therefore cannot overlap the anchor. If this first not-yet-ended target interval starts before $b$, then it overlaps $[a,b)$ and the anchor is returned. If it starts at or after $b$, then every later target interval starts no earlier and therefore also cannot overlap $[a,b)$. Thus the test is both sound and complete.

Algorithms~\ref{alg:buildFCFC}-\ref{alg:queryFCFC} implement the FCFC build and co-occurrence query; Figure~\ref{fcfc_example} shows the construction and a worked example with three events and nine intervals.
\begin{figure}
  \centering
  \begin{tikzpicture}[
    font=\scriptsize,
    >=latex,
    NC/.style={draw=gray!60,fill=white,minimum width=0.435cm,minimum height=0.44cm,inner sep=0.8pt,anchor=center},
    GC/.style={draw=green!50!black,fill=green!20,minimum width=0.435cm,minimum height=0.44cm,inner sep=0.8pt,anchor=center},
    BC/.style={draw=blue!60,fill=blue!15,minimum width=0.435cm,minimum height=0.44cm,inner sep=0.8pt,anchor=center},
    RC/.style={draw=red!50,fill=red!10,minimum width=0.435cm,minimum height=0.44cm,inner sep=0.8pt,anchor=center},
  ]
  \def\XO{1.15}    
  \def\CW{0.435}   

  \node[anchor=east,font=\small\bfseries] at (\XO-0.05, 0.22) {$L[0]$};
  \node[anchor=west,font=\small] at (\XO+0.05, 0.22)
    {$[\textcolor{green!50!black}{4},\textcolor{blue!60}{8}]\ \
     [\textcolor{green!50!black}{11},\textcolor{blue!60}{14}]\ \
     [\textcolor{green!50!black}{18},\textcolor{blue!60}{22}]$};
  \node[anchor=east,font=\small\bfseries] at (\XO-0.05,-0.33) {$L[1]$};
  \node[anchor=west,font=\small] at (\XO+0.05,-0.33)
    {$[1,3]\ [5,6]\ [13,15]\ [19,21]$};
  \node[anchor=east,font=\small\bfseries] at (\XO-0.05,-0.88) {$L[2]$};
  \node[anchor=west,font=\small] at (\XO+0.05,-0.88)
    {$[6,9]\ [16,17]$};

  \node[anchor=east,font=\small\itshape] at (\XO-0.05,-1.80) {Index};
  \node[anchor=east,font=\small\itshape] at (\XO-0.05,-2.26) {Time};
  \foreach \ii/\tv in {
      0/1,1/3,2/4,3/5,4/6,5/7,6/8,7/9,
      8/11,9/13,10/14,11/15,12/16,13/17,14/18,15/19,16/21,17/22}{
    \pgfmathsetmacro\XC{\XO+(\ii+0.5)*\CW}
    \def\sty{NC}%
    \ifnum\ii=2  \def\sty{GC}\fi%
    \ifnum\ii=8  \def\sty{GC}\fi%
    \ifnum\ii=14 \def\sty{GC}\fi%
    \ifnum\ii=6  \def\sty{BC}\fi%
    \ifnum\ii=10 \def\sty{BC}\fi%
    \ifnum\ii=17 \def\sty{BC}\fi%
    \node[\sty] at (\XC,-1.80) {\ii};%
    \node[\sty] at (\XC,-2.26) {\tv};%
  }

  \node[anchor=east,font=\small\bfseries] at (\XO-0.05,-3.26) {$FC[0]$};
  \pgfmathsetmacro\Xtwo  {\XO+(2 +0.5)*\CW}
  \pgfmathsetmacro\Xeight{\XO+(8 +0.5)*\CW}
  \pgfmathsetmacro\Xfourt{\XO+(14+0.5)*\CW}
  \pgfmathsetmacro\Xsix  {\XO+(6 +0.5)*\CW}
  \pgfmathsetmacro\Xten  {\XO+(10+0.5)*\CW}
  \pgfmathsetmacro\Xsvnt {\XO+(17+0.5)*\CW}
  \node[GC,text=green!50!black,font=\bfseries] (sAt) at (\Xtwo,  -3.04) {4};
  \node[GC] at (\Xtwo,  -3.50) {2};
  \node[GC,text=green!50!black,font=\bfseries] (sBt) at (\Xeight,-3.04) {11};
  \node[GC] at (\Xeight,-3.50) {8};
  \node[GC,text=green!50!black,font=\bfseries] (sCt) at (\Xfourt,-3.04) {18};
  \node[GC] at (\Xfourt,-3.50) {14};
  \node[BC,text=blue!60,font=\bfseries] (eAt) at (\Xsix, -3.04) {8};
  \node[BC] at (\Xsix, -3.50) {6};
  \node[BC,text=blue!60,font=\bfseries] (eBt) at (\Xten, -3.04) {14};
  \node[BC] at (\Xten, -3.50) {10};
  \node[BC,text=blue!60,font=\bfseries] (eCt) at (\Xsvnt,-3.04) {22};
  \node[BC] at (\Xsvnt,-3.50) {17};
  \draw[->,gray!60,thick] (sAt.north) to[out=75,in=105] (eAt.north);
  \draw[->,gray!60,thick] (sBt.north) to[out=75,in=105] (eBt.north);
  \draw[->,gray!60,thick] (sCt.north) to[out=65,in=115] (eCt.north);

  \node[anchor=east,font=\small\bfseries] at (\XO-0.05,-4.63) {$FC[1]$};
  \foreach \ii/\dist in {
      0/0,1/2,2/1,3/0,4/7,5/6,6/5,7/4,
      8/2,9/0,10/0,11/4,12/3,13/2,14/1,15/0,16/{-1},17/{-1}}{
    \pgfmathsetmacro\XC{\XO+(\ii+0.5)*\CW}
    \def\sty{NC}%
    \ifnum\ii=2  \def\sty{GC}\fi%
    \ifnum\ii=8  \def\sty{GC}\fi%
    \ifnum\ii=14 \def\sty{GC}\fi%
    \node[\sty] at (\XC,-4.40) {\ii};%
    \node[\sty] at (\XC,-4.86) {\dist};%
  }
  \foreach \ii in {2,8,14}{
    \pgfmathsetmacro\XC{\XO+(\ii+0.5)*\CW}
    \node[green!50!black,font=\normalsize] at (\XC,-5.20) {$\checkmark$};%
  }

  \node[anchor=east,font=\small\bfseries] at (\XO-0.05,-6.02) {$FC[2]$};
  \foreach \ii/\dist in {
      0/5,1/3,2/2,3/1,4/0,5/0,6/0,7/7,
      8/5,9/3,10/2,11/1,12/0,13/{-1},14/{-1},15/{-1},16/{-1},17/{-1}}{
    \pgfmathsetmacro\XC{\XO+(\ii+0.5)*\CW}
    \def\sty{NC}%
    \ifnum\ii=2  \def\sty{GC}\fi%
    \ifnum\ii=8  \def\sty{RC}\fi%
    \ifnum\ii=14 \def\sty{RC}\fi%
    \node[\sty] at (\XC,-5.78) {\ii};%
    \node[\sty] at (\XC,-6.24) {\dist};%
  }
  \node[green!50!black,font=\normalsize] at (\Xtwo,  -6.60) {$\checkmark$};
  \node[red!70!black,  font=\normalsize] at (\Xeight,-6.60) {$\times$};
  \node[red!70!black,  font=\normalsize] at (\Xfourt,-6.60) {$\times$};

  \foreach \ii in {2,8,14}{
    \pgfmathsetmacro\XCl{\XO+(\ii+0.5)*\CW-0.08}
    \pgfmathsetmacro\XCr{\XO+(\ii+0.5)*\CW+0.08}
    \draw[->,thin,gray!70] (\XCl,-3.72) -- (\XCl,-4.16);%
    \draw[->,thin,gray!70,dashed] (\XCr,-3.72) -- (\XCr,-5.54);%
  }
  \end{tikzpicture}
  \captionof{figure}{\small An example to illustrate the construction of FCFC and the query to find all ``Co-occurrence'' patterns between $L[0]$ and $L[1]$, and $L[0]$ and $L[2]$}
  \label{fcfc_example}
\end{figure}
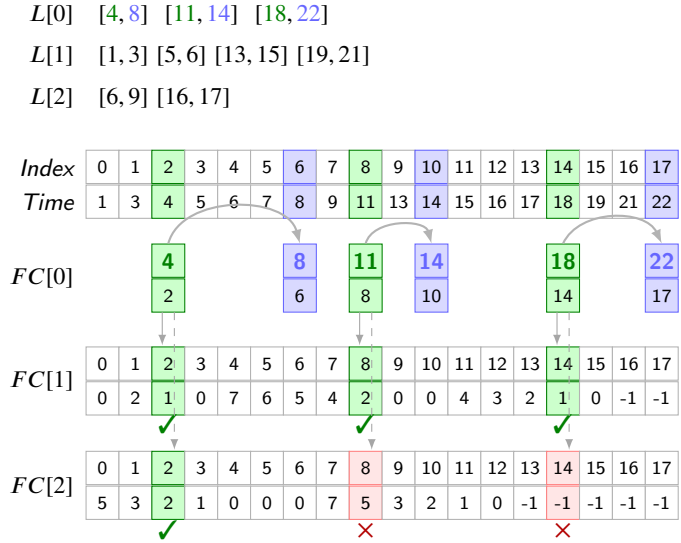

\paragraph{Complexity analysis.}
Consider a dataset $L$ with $m$ event types and $l$ total intervals. The global index contains at most $2l$ boundary points. FCFC stores two index arrays per event type over this global boundary list, so the asymptotic storage is $O(ml)$ and the construction time is $O(ml)$ after the global boundary list has been sorted. For a pattern-matching query with $a$ anchor intervals and $k$ target event types, each anchor-target check is a constant-time lookup followed by one interval-boundary comparison. The query time is therefore $O(ka)$, or $O(k l/m)$ under a balanced-event assumption. In practice, choosing the rarest clinically relevant event as the anchor often decreases $a$ substantially. An NFA-style alternative that repeatedly flattens, sorts, and scans the relevant interval streams requires $O(l' \log l' + l')$ time per query, where $l'$ is the number of intervals participating in the pattern. FCFC avoids this repeated flattening and sorting while preserving the interval semantics needed by the QEL/BEST model-checking layer.

\section{Results}
\label{sec:results}

\subsection{Implementation and Experimental Setup}
We implemented 2DFC and FCFC in Python 3.11 as execution components under the QEL/BEST model-checking layer. We also implemented 2DRT and RTFC for comparison. Build times for 2DRT and RTFC include the required flattening or restructuring of multi-stream input; 2DFC takes the raw sorted interval lists directly, exploiting the BEST invariant. FCFC precomputes event-distance arrays over a global boundary index. Persistent structures are backed by MongoDB for large-scale testing and deployment compatibility with NSRR-style data services.

\subsection{Dataset}
\label{sec:exp_dataset}
Synthetic datasets consist of $m$ events with $n'$ non-overlapping intervals each, generated uniformly on $[0,200{,}000{,}000]$ with max length 100. Three dataset collections vary: (1) event count with constant intervals/event; (2) event count with constant total intervals; and (3) intervals/event with constant event count.

CCSHS dataset within NSRR was used for real-world evaluation because it is a population-based pediatric cohort with standardized expert scoring, 23 event labels, and over 200k non-overlapping annotation intervals, providing rich sleep-stage, respiratory, desaturation, arousal, and limb-movement information suitable for QEL/BEST temporal phenotypes and the 2DFC/FCFC index structures. While all NSRR datasets share a common sleep data format and could in principle be queried by the same framework, they target different cohorts and study designs; adding multiple datasets would not substantively change the behavior of the temporal query patterns evaluated here and a single, well-characterized cohort is straightforward to scale up synthetically to stress-test performance without introducing additional cohort-specific variability.

\subsection{2DFC query performance}
We first evaluated the data structure build time of 2DFC compared to RTFC and 2DRT on synthetic datasets (Table~\ref{tab:result_build}). Because the interval lists in an annotated PSG are naturally pre-sorted (see Definition~\ref{def:interval_ensemble}), 2DFC takes advantage of this property to build in linear $O(n)$ time and $O(n)$ space. In contrast, 2DRT and RTFC require $O(n\log n)$ time. Consequently, 2DFC represents the fastest indexing approach for datasets up to 50 million intervals; as datasets grow larger, RTFC gradually overtakes 2DRT while 2DFC remains significantly more efficient than both.

\begin{table}[h]
  \begin{center}
  \caption{Data structure build time in seconds for 2DFC, RTFC, and 2DRT on different sizes of synthetic data.}
  \begin{tabular}{| r | r | r | r | r | } \hline
  {\bf Num of Events } & {\bf Total Num of Intervals} & {\bf 2DFC } & {\bf RTFC} & {\bf 2DRT} \\ \hline
  200 & 2,000,000 & 20 & 83 & 49 \\ \hline
  1,000 & 2,000,000 & 19 & 90 & 51 \\ \hline
  200 & 10,000,000 & 116 & 476 & 295 \\ \hline
  200 & 50,000,000 & 1,205 & 4,466 & 3,475 \\ \hline
  200 & 90,000,000 & 3,655 & 11,549 & 23,902 \\ \hline
  \end{tabular}
  \label{tab:result_build}
\end{center}
\end{table}

\begin{figure}
  \centering
  \subfloat[Events, constant intervals/event.]{%
    \includegraphics[width=0.31\linewidth]{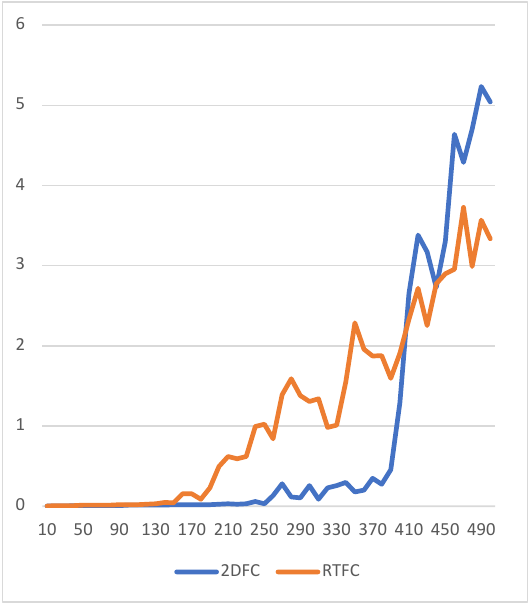}%
    \label{fig:result_event_number2}}
  \hfill
  \subfloat[Events, constant total intervals.]{%
    \includegraphics[width=0.31\linewidth]{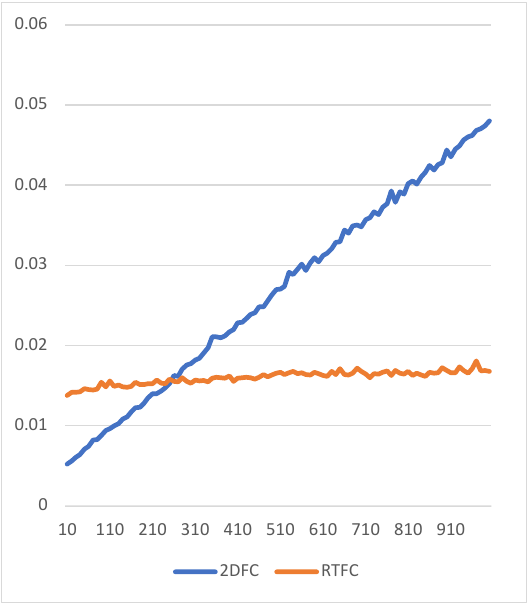}%
    \label{fig:result_event_number1}}
  \hfill
  \subfloat[Intervals/event, constant events.]{%
    \includegraphics[width=0.31\linewidth]{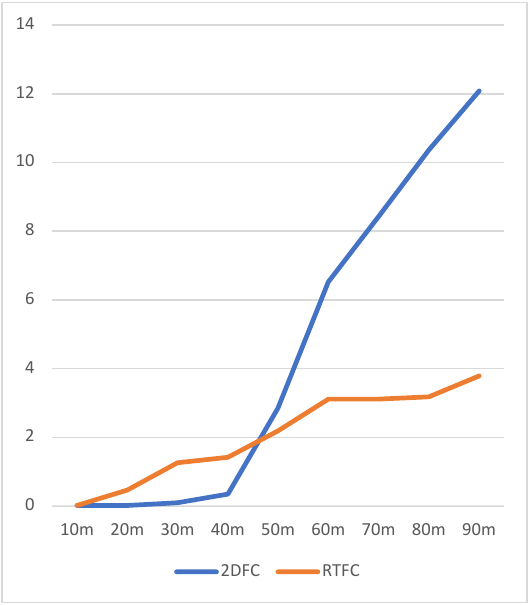}%
    \label{fig:result_interval_number}}\\
  \subfloat[Different query start points.]{%
    \includegraphics[width=0.44\linewidth]{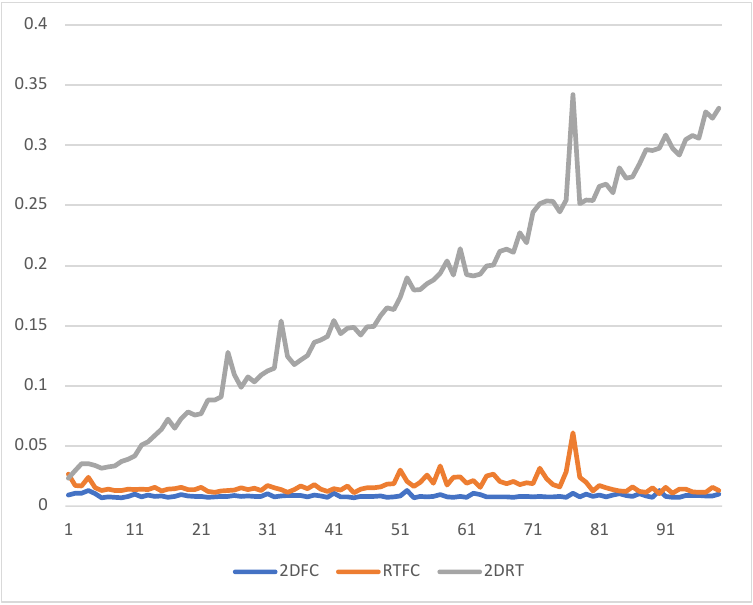}%
    \label{fig:result_query_position}}
  \hfill
  \subfloat[Different query ranges.]{%
    \includegraphics[width=0.44\linewidth]{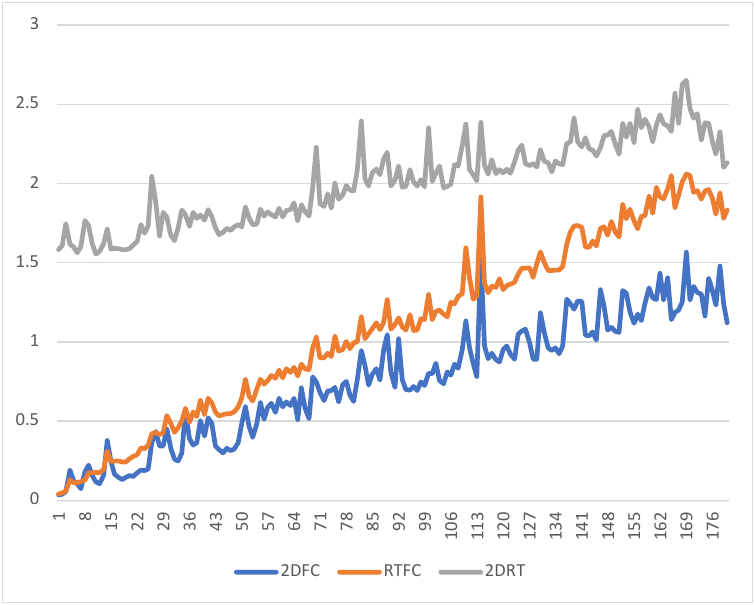}%
    \label{fig:result_query_range}}
  \caption{\small Query performance of 2DFC, RTFC, and 2DRT. (a)-(c) vary event count and interval count; (d)-(e) vary query interval start point and range.}
  \label{fig:result_query_performance}
\end{figure}

To evaluate query execution scalability, we performed three experiments varying the event and interval characteristics of the data. Due to its slow execution speeds at scale, 2DRT was excluded from these tests. 

First, holding the number of intervals per event constant at 100,000 while increasing the event count from 10 to 500 (Figure~\ref{fig:result_event_number2}), 2DFC demonstrated overall lower latency and better scalability for workloads with fewer than 400 events. Second, we fixed the total number of intervals at 2,000,000 and distributed them across 10 to 1,000 events (Figure~\ref{fig:result_event_number1}). Here, 2DFC query time scaled linearly with the event count, though RTFC outperformed 2DFC when querying across smaller event sets (fewer than approximately 210 events). Finally, holding the event count steady at 200 while scaling the total intervals from 10 million to 90 million (Figure~\ref{fig:result_interval_number}), RTFC exhibited more favorable scaling compared to 2DFC. This behavior illustrates the expected algorithmic tradeoff: 2DFC is subject to an additive cost per event type, whereas RTFC relies on a logarithmic range-tree traversal that acts efficiently when pure interval volume dominates.

To assess the sensitivity of the index structures to query parameters, we conducted two additional experiments on smaller fixed datasets, this time including 2DRT for comprehensive comparison. 
First, we tested sensitivity to the query's temporal location, executing 1,000 queries with varying start points over a dataset of 200 events and 100,000 intervals per event (Figure~\ref{fig:result_query_position}). While 2DRT query time scaled linearly as the start point shifted later in time, 2DFC maintained a constant performance independent of the starting temporal boundary. Second, we assessed sensitivity to varying query ranges lengths (Figure~\ref{fig:result_query_range}) using a dataset of 100 events and 20,000 intervals per event. For query lengths spanning 1 million to 90 million units (returning 2,000 to 180,000 results), all internal structures showed linear scaling. In these comparisons, 2DFC consistently provided the lowest query time overall compared to RTFC and 2DRT, with RTFC demonstrating the steepest latency slope.

\begin{figure}
  \centering
  \subfloat[RAM]{
    \includegraphics[width=0.48\textwidth]{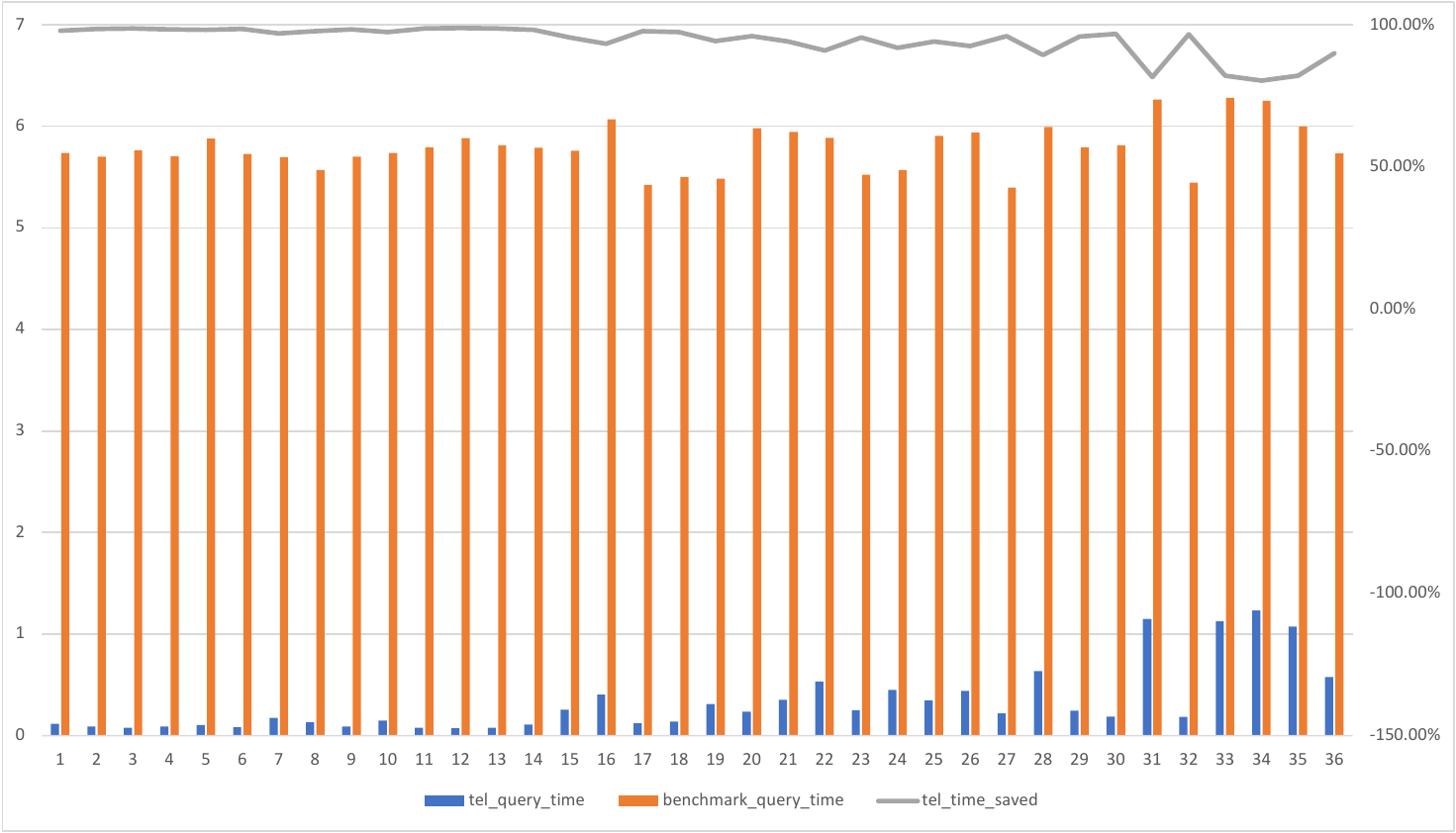}
    \label{fig:fcfc_ram}
  }
  \hfill
  \subfloat[MongoDB]{
    \includegraphics[width=0.48\textwidth]{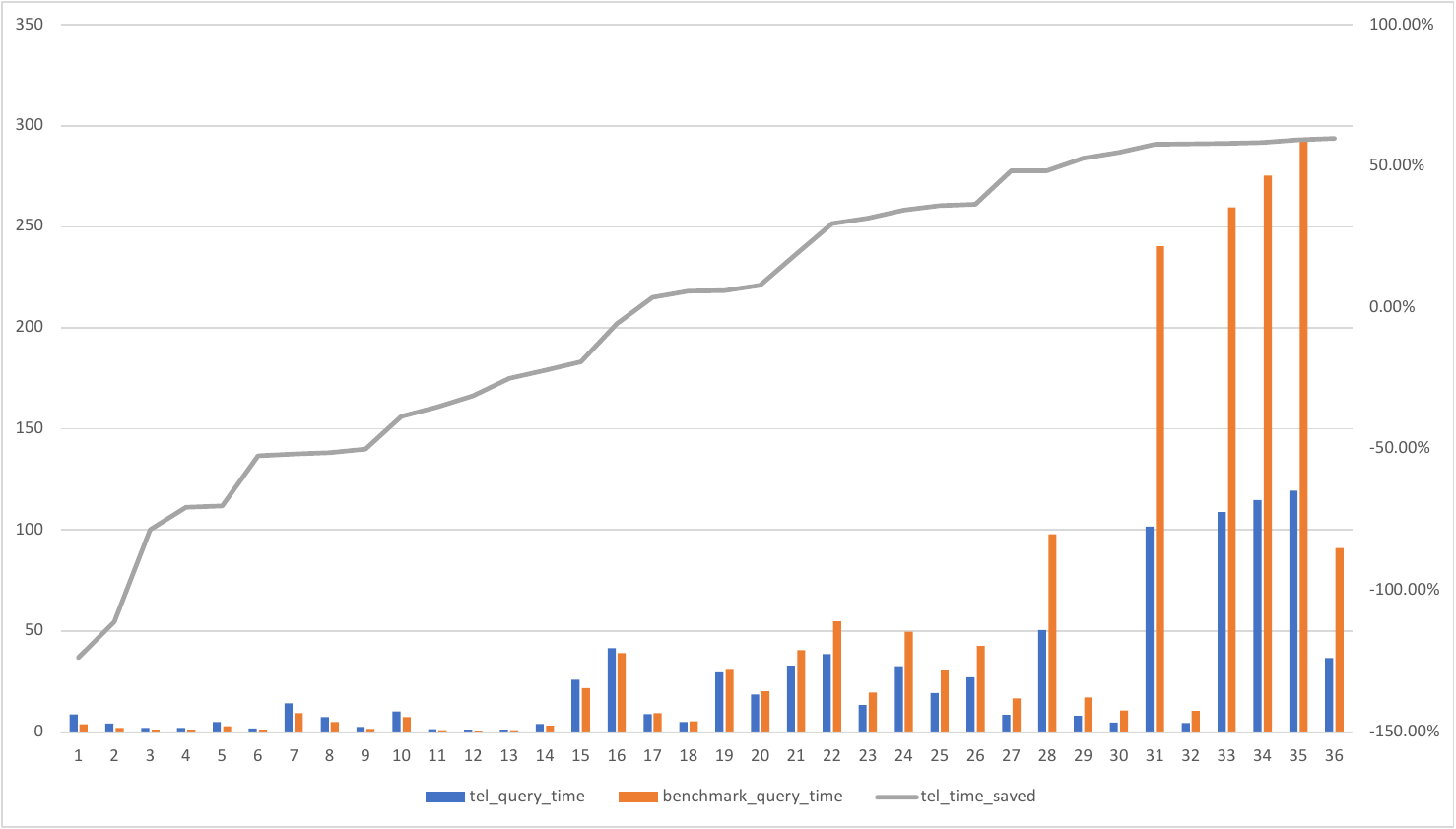}
    \label{fig:fcfc_mongo}
  }
  \caption{Dual-event query performance comparison with 36 testing queries (x axis) on CCSHS dataset. Blue bars represent the ``Before variant 2'' query time using FCFC, and orange bars represent the benchmark query time. Query time is in seconds (left y axis) The gray line represents the query time saved by FCFC in percentage (right y axis). (a) shows the results on RAM, and (b) shows the results on MongoDB.}
  \label{fig:tel_matching_performance}
\end{figure}

\subsection{FCFC query performance}

While 2DFC effectively handles multi-event overlap retrieval natively via bounded queries, the FCFC variant is designed specifically to accelerate sequential dual-event temporal matching. We evaluated FCFC query performance based on 36 dual-event ``Before variant 2'' queries (using a fixed anchor event $A$, 36 distinct target events $B$, and a threshold of $m=600$ time units). We benchmarked FCFC against a conventional NFA-style approach which utilizes a repeated flatten-sort-scan algorithm, testing both over strictly in-memory (RAM) and database-backed (MongoDB) execution environments.

As illustrated in Figure~\ref{fig:tel_matching_performance}, the algorithm yielded significant practical benefits. In RAM implementations, FCFC consistently reduced query execution times by 80-98\% by bypassing repeated time comparisons, instead acting upon computationally inexpensive precomputed distance arrays. In MongoDB scenarios, both methods encountered generally higher latencies due to fundamental database I/O overhead. With very low-frequency occurrences, this underlying I/O cost heavily dominated execution time and narrowed the relative gap between the implementations. However, for heavily-represented queries yielding substantial candidate sets, which is usually the cause of temporal processing bottlenecks, FCFC successfully bridged this gap to achieve around a 50\% overall query time reduction.

\subsection{QEL query application on real-world data}

To illustrate the clinical utility of the QEL/BEST framework, we first consider cohort selection for established studies of central and obstructive sleep apnea, where differentiating central from obstructive events is essential for characterizing distinct pathophysiological phenotypes~\cite{donovan2016prevalence}. Using QEL, analogous cohort inclusion criteria can be expressed directly at the event-logical level.

For central apnea, a basic QEL cohort selection query is
\[
  \left( \mdiamond_{m} \text{CA} \right)_0 \, .
\]
\textit{\textbf{Translation:} ``For this subject, does the overnight PSG contain at least one central apnea (CA) event anywhere in the recording window \( [0,m) \)''}
Evaluating this formula on the CCSHS cohort returns 285 subjects (55.3\% of the sample) in 0.06 seconds. A structurally identical single-event query for obstructive apnea is
\[
  \left( \mdiamond_{m} \text{OA} \right)_0 \, .
\]

In pediatric populations, the onset of desaturation following apnea can be extremely rapid with a median delay of 12 seconds, carrying significant clinical implications~\cite{patel1994age}. Related work also demonstrated tight coupling between apnea, desaturation, and bradycardia in preterm infants~\cite{poets1993bradycardia}. To explore such temporal relationships in CCSHS, we formulate a preliminary cohort-selection-style query capturing subjects in whom a central apnea precedes an SpO\textsubscript{2} desaturation using the \textbf{Before variant 2} template:
\[
  \exists q \exists x \exists y 
  \Bigl(
    \boldBox_x \neg\text{CA}
    \;\land\;
    \bigl(\boldBox_y (\text{CA} \land \neg\text{SPO2\_DESAT})\bigr)_{x}
    \;\land\;
    (\mdiamond_{600} \text{SPO2\_DESAT})_{x+y}
  \Bigr)_q \, .
\]
\textit{\textbf{Translation:} ``Does this subject exhibit at least one episode in which a central apnea (CA) begins after a CA-free baseline, persists for \(y\) seconds without overlapping SpO\textsubscript{2} desaturation, and is followed within 600 seconds by an SpO\textsubscript{2} desaturation event-mirroring analyses that select infants or children whose desaturations are temporally attributable to preceding apneas?''}
This query searches for sequences in which a central apnea clearly precedes an SpO\textsubscript{2} desaturation, separating desaturation dynamics caused by apnea from those driven by other pathophysiology (e.g., intrinsic lung disease). When evaluated across all 515 CCSHS subjects, 254 subjects satisfy this desaturation-to-apnea pattern (Row~4 of Table~\ref{tab:ccshs_performance}). By adjusting the temporal bounds (e.g., varying the maximum allowable delay), the same QEL template can be used to compare desaturation-latency distributions across subgroups, analogous to prior work quantifying apnea-desaturation delays as physiological markers.

We assessed query speed for a representative set of QEL formulas, comprising one single-event query, five dual-event queries, and three event-extraction queries, on the CCSHS dataset described in Section~\ref{sec:exp_dataset}. The evaluated queries, together with their template names (see Section~\ref{sec:template}) and constituent events, are:
\begin{itemize}
  \item (1) Anywhere-on-timeline existence: central apnea.
  \item (2) Co-occurrence: central apnea and sleep stage N3.
  \item (3-6) Before variants 1-4: central apnea and SpO\textsubscript{2} desaturation.
  \item (7) Clock-anchored existence extraction: events between 3:00 and 4:00 AM.
  \item (8) Event-anchored existence extraction: windows centered on central apnea.
  \item (9) Event-anchored universality extraction: regions in which central apnea persists.
\end{itemize}

The query performance results are summarized in Table~\ref{tab:ccshs_performance}. The \texttt{x1}, \texttt{x10}, \texttt{x100}, and \texttt{x1000} columns correspond to the original CCSHS dataset and synthetically scaled datasets containing 10, 100, and 1{,}000 times as many subjects, respectively. Execution times are reported as averages over 10 runs per condition. Across all scales, the QEL engine efficiently handles cohort-selection and event-extraction workloads on large sleep datasets: even at \texttt{x1000}, cohort-selection-style queries execute in under 45 seconds. For extraction queries that return large event sets, runtime growth is dominated by memory and data-transfer overhead rather than by the temporal-logic model-checking semantics themselves, consistent with experience from large-scale oximetry-based screening studies~\cite{alvarez2007ctm}.

\begin{table}[ht]
  \begin{center}
  \caption{Query performance on CCSHS dataset at different scales. Execution time (Exe time) is in seconds and reported as average over 10 runs. $*$The query only returns the subjects satisfying the query condition, so the number of events is not applicable.}
  \begin{tabular}{|l|r|r|r|r|r|r|} \hline
  {\bf} & {\bf Num of} & {\bf Num of} & {\bf Exe time} & {\bf Exe time} & {\bf Exe time} & {\bf Exe time} \\
  & {\bf subjects} & {\bf events} & {\bf (x1)} & {\bf (x10)} & {\bf (x100)} & {\bf (x1000)} \\ \hline
  1 & 285 & $N/A^*$& 0.06 & 0.09 & 0.46 & 4.18 \\ \hline
  2 & 285 & $N/A^*$& 0.13 & 0.91 & 4.25 & 42.19 \\ \hline
  3 & 224 & $N/A^*$& 0.10 & 0.43 & 4.19 & 44.95 \\ \hline
  4 & 254 & $N/A^*$& 0.11 & 0.53 & 4.22 & 42.06 \\ \hline
  5 & 235 & $N/A^*$& 0.10 & 0.51 & 4.09 & 40.91 \\ \hline
  6 & 254 & $N/A^*$& 0.15 & 0.51 & 4.11 & 43.93 \\ \hline
  7 & 515 & 21,667 & 0.45 & 1.97 & 7.91 & 921.45 \\ \hline
  8 & 285 & 2,700  & 0.09 & 0.23 & 1.70 & 76.16 \\ \hline
  9 & 285 & 1,747  & 0.07 & 0.20 & 1.60 & 72.78 \\ \hline
  \end{tabular}
  \label{tab:ccshs_performance}
\end{center}
\end{table}

To demonstrate the expressive power of QEL beyond standard cohort templates, we next present two illustrative patterns that align with complex clinical phenotyping tasks in the sleep-medicine literature.

\paragraph{Example 1: Pre-apnea medication washout and sleep-stage restriction.}

In interventional sleep trials, it is common to restrict analyses to epochs with/without confounding medications or interventions and to focus on specific sleep stages where respiratory events have clearer physiological interpretation. For instance, we can encode a medication washout pattern of a desaturation-to-apnea delay with medication and sleep-stage constraints:
\[
  \exists q
  \Bigl(
    \boldBox_{600} \neg(\text{MED} \lor \text{CA} \lor \text{OA})
    \;\land\; 
    \bigl(\boldBox_{5} (\text{CA} \lor \text{OA}) \land \text{N3} \land \neg\text{SPO2\_DESAT}\bigr)_{600} 
    \;\land\; 
    (\mdiamond_{20} \text{SPO2\_DESAT})_{605} 
  \Bigr)_q \, .
\]
\textit{\textbf{Translation:} ``Does this subject have at least one medication-free baseline interval without central or obstructive apnea for 600 seconds, followed by a central or obstructive apnea occurring in N3 sleep that lasts at least 5 seconds without overlapping SpO\textsubscript{2} desaturation, and then an SpO\textsubscript{2} desaturation beginning within 20 seconds?''}
This example shows how QEL can simultaneously encode washout windows, stage-specific constraints, and apnea-desaturation timing, providing a declarative analogue of complex inclusion criteria used in mechanistic and interventional studies of apnea.

\paragraph{Example 2: Signal-level physiological event definitions.}

In many large datasets, respiratory and hypoxemic events are available only as precomputed labels, which can suffer from inter-scorer variability and evolving scoring rules~\cite{aasm2012rules}. QEL instead allows event definitions to be specified directly from continuous physiological signals, in line with the American Academy of Sleep Medicine (AASM) scoring criteria and signal-level approaches to oximetry-based screening~\cite{aasm2012rules,alvarez2007ctm}. We illustrate this with two atomic constructs:

A signal-level central apnea can be captured in QEL as
\[
  \varphi_{\text{CA\_signal}} = \exists d \geq 10 \boldBox_{d} \bigl(\ell_{\text{AirflowCessation}} \land \ell_{\text{EffortAbsence}}\bigr) \, .
\]
\textit{\textbf{Definition:} ``An event is classified as a central apnea if: (i) the airflow signal shows a complete or near-complete cessation (nasal cannula pressure or thermistor amplitude close to zero, encoded by atomic label \( \ell_{\text{AirflowCessation}} \)) lasting at least 10 seconds; and (ii) respiratory effort signals show simultaneous absence of chest and abdominal movements (plethysmography bands flat, encoded by atomic label \( \ell_{\text{EffortAbsence}} \)) throughout the event, consistent with AASM criteria for central apnea~\cite{aasm2012rules}.''}

Let stable pre-event baseline window be 30 seconds. A signal-level desaturation event is defined as
\[
  \varphi_{\text{Desat\_signal}} = \exists d \geq 10 \; \bigl(\boldBox_{30} \ell_{\text{SpO2$>94$\%}}\bigr) \;\land\; \bigl(\boldBox_d \ell_{\text{SpO2Drop$\geq3$\%}}\bigr)_{30} \, .
\]
\textit{\textbf{Definition:} ``A hypoxemic event is defined as a transient drop in the SpO\textsubscript{2} signal of at least 3 to 4\% relative to a stable pre-event baseline \(> 94\%\) lasting at least 10 seconds.''}

We can express the desaturation-to-CA delay cohort selection query using signal-level definitions as:
\begin{align}
  &\exists q \; \exists d \le 12 \; \exists x \geq 10 \; \exists y \geq 10 \notag \\
  &\quad \bigl(\boldBox_{30} \bigl( \ell_{\text{SpO2${>}94$\%}} 
      \land \neg(\ell_{\text{AirflowCessation}} \land \ell_{\text{EffortAbsence}})\bigr) \notag \\
  &\quad {}\land\; \bigl( \boldBox_{x} (\ell_{\text{AirflowCessation}} 
      \land \ell_{\text{EffortAbsence}}) \land \boldBox_{d} \neg \ell_{\text{SpO2Drop$\geq3$\%}}\bigr)_{30} \notag \\
  &\quad {}\land\; \bigl(\boldBox_{y} \ell_{\text{SpO2Drop$\geq3$\%}}\bigr)_{30+d} \bigr)_{q} \, .
\end{align}

Because all components are defined directly on raw signals, QEL can reproduce and extend signal-processing pipelines used in apnea-hypopnea research while retaining a single declarative language for both event definition and cohort selection.

\section{Discussion}
\label{sec:discussion}

\paragraph{Symbolic biomedical informatics contribution.}
This logic-based framework recasts sleep temporal querying as a data-management and secondary-use problem rather than a pure algorithmic range-search problem. QEL formulas provide explicit, reusable phenotype definitions; BEST provides a common interval-ensemble representation for NSRR-style annotations; and model checking provides a precise semantics for cohort inclusion. Unlike ad hoc scripts, a phenotype or cohort expressed as a QEL formula is fully reproducible, straightforward for investigators to read, and readily processed by machines.

\paragraph{Model-checking rigor.}
The finite-evidence theorem in Section~\ref{sec:indexing} guarantees that QEL queries over dense rational time can be evaluated by finitely many representative checks once a BEST and a formula are fixed. Thus the endpoint-array representation is a sound reduction from the formal semantics. This is especially important for real-world clinical research, where small differences in event anchoring, delay windows, and inclusion/exclusion criteria can materially alter cohort membership.

\paragraph{2DFC build time and space efficiency.}
2DFC is significant because it exploits a biomedical-data property that generic range structures ignore: per-label sleep annotation intervals are already sorted and non-overlapping after BEST normalization. At 90M intervals, 2DFC takes 3,655 s compared with 11,549 s for RTFC and 23,902 s for 2DRT, yielding approximately $3\times$ and $6\times$ reductions (Table~\ref{tab:result_build}). Its $O(n)$ build time and $O(n)$ space are important for practical repository services, where indexes must be rebuilt or refreshed as datasets are curated and shared.

\paragraph{2DFC query performance and tradeoffs.}
For typical sleep datasets ($m\approx23$, $n>200{,}000$), the $O(m)$ term in 2DFC's $O(\log|FC[0]|+m+k)$ complexity is negligible. At larger $m$, RTFC ($O(\log n+k)$) can overtake 2DFC; in our synthetic experiments, the crossover occurs near $m=210$ at 2M total intervals (Figure~\ref{fig:result_event_number1}). In practice, event lists can be ordered by ascending cardinality and compatible non-overlapping labels can be merged to reduce the additive event-type term.

\paragraph{FCFC performance for dual-event pattern matching.}
FCFC targets a different part of the model-checking workload: dual-event relationships such as co-occurrence, before/after variants, and bounded delay. By replacing repeated flatten-sort-scan operations with precomputed distance arrays, FCFC achieves 80-98\% query-time reduction versus the NFA-style baseline in RAM mode and approximately 50\% reduction in MongoDB mode for high-candidate queries (Figure~\ref{fig:tel_matching_performance}). Choosing the rarest event as the anchor further improves throughput.

\paragraph{Real-world scalability.}
On CCSHS (515 subjects, 202,587 intervals, 23 event labels), all evaluated queries run at sub-second latency at original scale. Cohort-selection queries (Table~\ref{tab:ccshs_performance}, Q1-6) complete within 45 s at $1{,}000\times$ scale. Event-extraction query Q7 reaches 921 s at $1{,}000\times$, reflecting result-set size and hardware constraints. These experiments show that formal model checking and scalable systems design can coexist in a deployable sleep-data platform.

\paragraph{Limitations and extensions.} 2DFC and FCFC assume normalized non-overlapping interval ensembles; overlapping annotations require preprocessing or label-specific normalization. FCFC's $O(ml)$ space can be large and may benefit from sparse, on-demand, or compressed construction. Fully general nested QEL formulas require recursive application of the finite-evidence theorem beyond the query templates evaluated here. 
The current framework is designed for retrospective, static datasets: adding new events or subjects requires rebuilding the 2DFC and FCFC index structures, and incremental or online index maintenance is an engineering extension for continuously growing repositories. 
The cohort-discovery layer currently relies on a fixed library of query templates; translation of arbitrary free-form QEL formulas into executable query statements, together with automated query optimization and plan selection, is an important direction for future work.
Future work should also integrate uncertainty in annotations, support cross-study harmonization of event vocabularies, and expose QEL formula libraries through data-sharing platforms so that temporal phenotypes can be reused across cohorts.

\section{Conclusion}
\label{sec:conclusion}

We presented a logic-based temporal cohort discovery engine for sleep-study data that integrates formal representation, model checking, index design, and empirical validation. QEL provides a dense-time logical foundation for expressing temporal phenotypes over PSG annotations and signal-derived intervals. BEST provides a compact subject-level representation in which every event label is interpreted as a finite interval ensemble. Together, QEL and BEST make cohort inclusion criteria explicit and reproducible by bridging human readability and machine computability: a subject satisfies a cohort definition precisely when the corresponding BEST model satisfies the QEL formula.

The indexing contributions make this formal semantics practical for large repositories. 2DFC exploits the sorted, non-overlapping structure of normalized sleep annotation intervals to support efficient interval-overlap retrieval with $O(n)$ build time and $O(n)$ space. FCFC targets the complementary workload of dual-event temporal pattern matching by linking boundaries to a global index and enabling $O(1)$ target-event lookup for each anchor-target comparison after preprocessing. Experiments on synthetic datasets up to 90 million intervals and on CCSHS annotations from NSRR demonstrate that this combination of formal model checking and specialized indexing can support both cohort-selection and event-extraction workloads at repository scale.

This work positions temporal cohort discovery as a first-class biomedical informatics capability rather than as a collection of ad hoc scripts. By organizing sleep research requirements into single-event retrieval, dual-event temporal pattern matching, and event data extraction, the framework provides reusable query templates that can encode clinically meaningful timing constraints, including AASM-inspired duration thresholds, co-occurrence, bounded delay, stage restriction, and absence windows. The same principles can be extended to other biomedical domains where high-resolution physiological time series and interval annotations are central to computational phenotyping.

\section*{Acknowledgment}\label{ack}

This work has been supported in part by the U.S. National Science Foundation Award IIS2500624 and the National Institutes of Health grants R01AG084236 and U24AG098157. The views expressed in this paper are those of the authors and do not necessarily reflect those of the funding agencies.

\section*{Declaration of Competing Interest}

The authors declare that they have no known competing financial interests or personal relationships that could have appeared to influence the work reported in this paper.

\bibliographystyle{cas-model2-names}

\bibliography{cas-refs}

\end{document}